\documentclass{IEEEtran}
\pdfoutput=1
\usepackage{cite}
\usepackage{color}
\usepackage{graphicx}
\usepackage{epstopdf}
\usepackage{subfig}
\usepackage[cmex10]{amsmath}
\usepackage{amssymb}
\usepackage{algorithm,algpseudocode}

\usepackage{amsmath,amssymb}
\usepackage{lipsum}
\usepackage{comment}
\usepackage{array}
\input{mysymbol.sty}
\usepackage{needspace}

\usepackage{amsthm,bm}
\usepackage{bbm}
\usepackage{url}

\newtheorem{proposition}{Proposition}

\newtheorem{theorem}{Theorem}
\newtheorem{definition}{Definition}
\newtheorem{lemma}{Lemma}

\theoremstyle{definition}
\newtheorem{assumption}{Assumption}

\author{Santiago Paternain$^\dagger$, Juan Andr\'es Bazerque$^*$, Austin Small$^\dagger$ and Alejandro Ribeiro$^\dagger$
\thanks{Work supported by ARL DCIST CRA W911NF-17-2-0181. $^\dagger$Department of Electrical and System Engineering, University of Pennsylvania. Email: \{spater, ausmall, aribeiro\}@seas.upenn.edu. $^*$Departamento de Ingenier\'ia El\'ectrica, Facultad de Ingenier\'ia, Universidad de la Rep\'ublica. Email: jbazerque@fing.edu.uy
}}

\renewcommand{\comment}[1]{}
\newcolumntype{S}{>{\centering\arraybackslash} m{.10\linewidth} }
\newcolumntype{T}{>{\centering\arraybackslash} m{.30\linewidth} }
\title{Stochastic Policy Gradient Ascent in Reproducing Kernel Hilbert Spaces}
\begin{document}

\maketitle


%
\begin{abstract}
Reinforcement learning consists of finding policies that maximize an expected cumulative long term reward in a Markov decision process with unknown transition probabilities and instantaneous rewards. In this paper we consider the problem of finding such optimal policies while assuming they are continuous functions belonging to a reproducing kernel Hilbert space (RKHS). To learn the optimal policy we introduce a stochastic policy gradient ascent algorithm with three unique novel features: (i) The stochastic estimates of policy gradients are unbiased. (ii) The variance of stochastic gradients is reduced drawing on ideas from numerical differentiation. (iii) Policy complexity is controlled using sparse RKHS representations. Novel feature (i) is instrumental in proving convergence to a stationary point of the expected cumulative reward. Novel feature (ii) facilitates reasonable convergence times. Novel feature (iii) is a necessity in practical implementations which we show can be done in a way that does not eliminate convergence guarantees. Numerical examples in standard problems illustrate successful learning of policies with low complexity representations which are close to stationary points of the expected cumulative reward.
\end{abstract}

%
\section{Introduction}

Markov decision processes (MDPs) \cite{howard1964dynamic} provide a mathematical framework for modeling decision making in situations where outcomes are partly random and partly under the control of a decision maker. This general framework has been used to study systems in diverse disciplines such as robotics \cite{kober2013reinforcement}, control \cite{shreve1978alternative}, and finance \cite{rasonyi2005utility}. An MDP is a memoryless discrete time stochastic control process, where the state of the system at the next time is a random variable, whose probability distribution depends on the current state and the action selected by the decision maker. The actions selected by the agent determine instantaneous rewards that can be aggregated over a trajectory to determine cumulative rewards. The instantaneous rewards depend on both the state and the actions and thus, the reward along a trajectory depends on the policy under which the actions are selected based on the current state. In that sense, cumulative rewards are a measure of the quality of the decision making policy, and the objective of the agent is to find a policy that maximizes the expectation of the cumulative rewards, which is known as the Q-function of the MDP \cite{sutton1998reinforcement}.

In this paper we consider reinforcement learning problems, in which the transition probabilities and the rewards are unknown and can only be accessed trough experiments that permit observation of realized transitions and rewards \cite{sutton1998reinforcement}. Solutions to these problems can be roughly divided among those that learn the Q-function to then chose for any given state the action that maximizes the function \cite{watkins1992q} and those that attempt to directly learn the optimal policy \cite{sutton2000policy, deisenroth2013survey}. Among the former, Q-learning is a standard solution that is applicable in discrete state and discrete action spaces \cite{watkins1992q}. A drawback of Q-learning, and any other algorithm that learns Q-functions for that matter, is that maximizing the Q-function to select optimal actions is itself computationally challenging. This motivates development of algorithms that attempt to learn the optimal policy directly by performing (stochastic) gradient ascent on the Q-function with respect to a policy variable \cite{sutton2000policy, deisenroth2013survey}.

The algorithms in \cite{watkins1992q, sutton2000policy, deisenroth2013survey} suffer from a dimensionality curse because the complexity of learning scales exponentially with the number of actions and states \cite{friedman2001elements}. This is of particular concern in continuous state-action spaces, where any reasonable discretization leads to a very large number of states and possible actions. As is the case in many other learning domains, a common approach to sidestep the dimensionality curse is to assume that either the Q-function or the policy admits a finite parametrization that can be linear \cite{sutton2009convergent}, rely on a nonlinear basis expansion \cite{bhatnagar2009convergent}, or be given by a neural network \cite{mnih2013playing}. Alternatively one can assume that the Q-function \cite{koppel2017breaking, tolstayanonparametric} or the policy \cite{lever2015modelling} belong to a reproducing kernel Hilbert space (RKHS) which provide the ability to approximate functions using nonparameteric functional representations. Although the structure of the space is determined by the choice of the kernel, the set of functions that can be represented is sufficiently rich to permit a good approximation of a large class of functions.  

Our focus here is on the convergence and complexity of policy learning in RKHS. The contributions of the paper are: 

\begin{itemize}
\item[(i)] We develop a method that computes unbiased stochastic policy gradients in the RKHS (Section \ref{sec_algorithm}). These unbiased estimates are plugged into a stochastic gradient ascent method that we can formally prove learns a policy that is a stationary point of the Q-function (Theorem \ref{thm_convergence_unbiased}). 
\item[(ii)] We introduce a mechanism to reduce the variance of the stochastic policy gradients thereby reducing the overall learning cost (Section \ref{sec_approximation}). 
\item[(iii)] We use sparse RKHS representations to learn policies of limited complexity (Section \ref{sec_komp}) that we can formally prove converge to a neighborhood of a stationary point of the Q-function (Theorem \ref{theo_parsimonious}). Numerical examples illustrate that RKHS policies of low complexity perform well in standard problems (Section \ref{sec_numerical}).
\end{itemize}

To produce unbiased estimates of policy gradients [cf. (i)] there are two main challenges that are addressed.
The first one arises from the fact that since the expression of the policy gradient depends on the Q-function itself, the Q-function has to be estimated. This can be solved using a stochastic estimator of said function (Algorithm \ref{alg_q_estimate}) that is unbiased (Proposition \ref{prop_unbiased_q}). The second difficulty when computing the gradient of the Q-function is that it depends on a state-action distribution that is not that of sample trajectories. Meaning that if one were to consider a trajectory of the system as a sample to compute the stochastic gradient, this estimate would be biased. This issue is typically reinforced by other policy gradient algorithms which consider a fixed horizon as an estimate of the infinite sequence of state and action pairs. The biases introduced by the mentioned algorithms prevent to show convergence of stochastic gradient ascent to a stationary point of the Q-function. To overcome these issues, we propose to use as stopping time a random variable drawn from a geometric distribution. Such stopping time defines a horizon that is representative of the infinite time horizon problem and hence yields an unbiased estimate (Proposition \ref{prop_unbiased_grad}). We emphasize that our policy gradient estimates are different form those in, e.g., \cite{lever2015modelling}, and that those differences are instrumental in proving convergence (Section \ref{sec_convergence}).

To reduce the variance of stochastic policy gradient estimates [cf. (ii)] we show that multiple samples from a Gaussian random policy can be related to numerical differentiation of the Q-function (Section \ref{sec_approximation}). This idea has been used in the zero-th order optimization literature \cite{nesterov2011random,ghadimi2013stochastic}. This is, when the gradient of the function one is trying to minimize cannot be directly computed, one can estimate it by considering random samples in a neighborhood of the iterate and evaluating the objective function at those points. 

The representations in Section \ref{sec_algorithm} have growing complexity because they require the addition of a kernel center and weight at each iteration of the stochastic gradient ascent algorithm. This memory explosion problem is typical of learning in RKHS and a major hurdle in practical implementation. Indeed, since we require as many kernel elements as stochastic gradient iterations we perform and convergence is asymptotic, we need, in principle, to add an infinite number of kernels to represent the optimal policy. Iterations are halted in practice but policy gradient nonetheless requires large number of iterations -- between $10^4$ and $10^5$ in the experiments in Section \ref{sec_numerical}. To control memory explosion of RKHS representations  [cf. (iii)] we follow the ideas in \cite{koppel2017breaking} to propose the use of orthogonal matching pursuit to construct sparse kernel representations (Section \ref{sec_komp}). By doing so, we ensure that the model order of the representation remains bounded for all iterates at the cost of achieving convergence only to a neighborhood of a critical point of the Q-function (Theorem \ref{theo_parsimonious}). The size of the neighborhood depends both on the learning rate -- step size-- selected and the error that one allows in the construction of sparse representations. Other than concluding remarks the paper ends with numerical experiments where we consider the mountain car and the cartpole problem (Section \ref{sec_numerical}). In both cases we successfully learn policies that are close to stationary points of the Q-function and that admit low complexity representations -- with 40 kernels for mountain car and 120 kernels for cartpole.



\section{Problem Formulation}
In this work we are interested in the problem of finding a policy that maximizes the expected reward of an agent that chooses actions sequentially. Formally, let us denote the time by $t \in \mathbb{N}$ and let $\ccalS $ be a compact set denoting the state space of the agent and $\ccalA = \mathbb{R}^p$ be its action space. The transition dynamics are governed by a conditional probability $ P_{s_t\to s_{t+1}}^{a_t}(s) :=p(s_{t+1}=s|(s_t,a_t) \in \ccalS\times \ccalA)$ satisfying the Markov property, i.e.,
%
$p(s_{t+1}=s\big|(s_u,a_u) \in \ccalS\times \ccalA, \forall u\leq t) 
    =p(s_{t+1}=s|(s_t,a_t) \in \ccalS\times \ccalA).$
%
    The policy of the agent is a map $h: \ccalS \to \ccalA$ and we assume it to be a vector-valued function in a vector-valued RKHS $\ccalH$. The reason for considering a vector-valued RKHS is that the system to be controlled might have more than one input. We formally define this notion next, and we relate it to the classic definition of  a scalar RKHS.
%
\begin{definition}\label{def_rkhs} 
  A vector valued RKHS $\ccalH$ is a Hilbert space of functions $h:\ccalS\to \mathbb{R}^p$ such that for all $\bbc\in\mathbb{R}^p$ and $\bbx \in \ccalS$, $
\left(\kappa_x\bbc\right)(\bby) = \kappa(\bbx,\bby)\bbc \in \ccalH  \quad\mbox{for all}\quad \bby \in \ccalS,$ 
  where $\kappa_x(\bby)$ is a symmetric function that is a positive definite matrix for any $\bbx,\bby \in \ccalS$ and it has the reproducing property 
  \begin{equation}\label{eqn_reproducing}
<h,\kappa_x\bbc>_{\ccalH} = h(\bbx)^\top\bbc.
  \end{equation}
  Without loss of generality we will assume that the Hilbert norm of $\kappa(\bbx,\cdot)$ is equal to one.
\end{definition}
%
%
If $\kappa(\mathbf x,\mathbf y)$ is a diagonal matrix-valued function with diagonal elements  $\kappa(\textbf x,\textbf y)_{ii}$,  and $\mathbf c$ is the $i$-th canonical vector in $\mathbb R^p$, then \eqref{eqn_reproducing} reduces to the standard one-dimensional reproducing property per coordinate 
%
  %
  %
$h_i(\bbx) = <h_i,\kappa(\bbx,\cdot)_{ii}>.$ 
  %
  %
%
The latter allows us to treat each individual input as an independent function in a RKHS.

  Instead of choosing the action deterministically as $a = h(s)$, we randomly draw it from a multivariate Gaussian distribution with mean $h(s)$. A random policy helps the exploration of the state space and it is a good approximation of the deterministic policy as we show in Proposition \ref{prop_approximation}. The conditional probability of the action is defined as $\pi_h(a|s):\ccalS\times \ccalA \to \mathbb{R}_+$, 
\begin{equation}
\pi_h(a|s) = \frac{1}{\det(2\pi\Sigma)}
             \exp\left[- \big(a-h(s)\big)^\top\Sigma^{-1}\big(a-h(s)\big)\right].
\end{equation}
The latter means that given a policy $h\in\ccalH$ and the current state $s\in\ccalS$, the agent selects an action $a\in\ccalA$ from a multivariate normal distribution $\ccalN(h(s),\Sigma)$. The actions selected by the agent result in a reward defined by a function $r:\ccalS\times\ccalA \to \mathbb{R}$. We assume these rewards to be uniformly bounded as we formally state next.
\begin{assumption}\label{assumption_reward_function}
There exists $B_r>0$ such that $\forall (s,a) \in \ccalS\times\ccalA$, the reward function $r(s,a)$ satisfies $|r(s,a)|\leq B_r$.
\end{assumption}

The objective is then to find a policy $h^* \in \ccalH$ such that it maximizes the expected discounted reward 
\begin{equation}\label{eqn_problem_statement}
  h^* := \argmax_{h\in\ccalH} U(h) = \argmax_{h\in\ccalH} \mathbb{E}\left[\sum_{t=0}^\infty \gamma^t r(s_t,a_t)\Big| h\right] , 
  \end{equation}
where the expectation is taken with respect to all states $s_0,s_1,\ldots $ and all actions $a_0,a_1,\ldots,$ and $\gamma \in(0,1)$ is a discount factor that gives relative weights to the reward at different times. Values of $\gamma$ close to one imply that rewards in the present are as important as future rewards, whereas smaller values of $\gamma$ give origin to myopic policies that prioritize maximizing immediate rewards. It is also noticeable that $U(h)$ is indeed a function of the policy $h$, since policies affect the joint probabilities of the trajectories $\{s_t,a_t\}_{t=0}^\infty$.

Conceivably, problem \eqref{eqn_problem_statement} could be solved iteratively by running a gradient ascent iteration on the space of functions. In parametric problems where variables lie in a finite space, gradient ascent converges to a critical point of $U(h)$ -- if $U(h)$ is upper bounded -- under constant and diminishing step size \cite[pp 43-45] {bertsekas1999nonlinear}. The same will be proved here in the case of maximizing a functional where the decision variable is a function in $\ccalH$. When the functional is a convex function these results have been established in \cite{koppel2016parsimonious,koppel2017decentralized}. 

The importance of this theoretical result notwithstanding, is limited by the computation of the gradient of $U(h)$ with respect to $h$ being intractable. To see why this is the case, define the discounted long-run probability distribution $\rho(s,a)$ 
\begin{align}\label{eqn_discounted_distribution}
      &\rho(s,a) := (1-\gamma)\sum_{t=0}^\infty \gamma^t p(s_t=s,a_t=a) 
    \end{align}
where $p(s_t=s,a_t=a)$ defines the probability of reaching state $s$ and action $a$ at time $t$, and is given by
		\begin{align}\label{eqn_def_pstat}
		& p(s_t,a_t)=\hspace{-.2cm}\int\hspace{-.1cm}\pi_h(a_t|s_t)\prod_{u=0}^{t-1}\hspace{-.1cm}p(s_{u+1}|s_u,a_u)\pi_h(a_u|s_u) p(s_0)d\bbs d\bba
    \end{align}
		and where $d\bbs=ds_0\ldots ds_{t-1}$ and $d\bba=da_0\ldots da_{t-1}$ imply integration over the previous states and actions.

Let $Q(s,a;h)$ be the expected discounted reward for a policy $h$ that at state $s$ selects action $a$, formally  defined as
\begin{equation}\label{eqn_q_function}
Q(s,a;h) := \mathbb{E}\left[\sum_{t=0}^\infty \gamma^t r(s_t,a_t)\Big| h, s_0=s, a_0=a\right].
\end{equation}
The previous functions are useful to write the expression of the gradient of $U(h)$ as we formally state in the next proposition.
\begin{proposition}[\cite{sutton2000policy,lever2015modelling}]
  The gradient of the discounted rewards with respect to $h$ yields  
\begin{align}\label{eqn_nabla_U}
     &\nabla_h U(h,\cdot)= 
\frac{1}{1-\gamma}\mathbb{E}_{\rho}\left[ Q(s,a;h)\kappa(s,\cdot)\Sigma^{-1}\left(a-h(s)\right) \Big| h\right],  \end{align}
where the dot in $(h,\cdot)$ substitutes the second variable of the kernel, belonging to $\mathcal S$,  which is omitted to simplify notation.
\end{proposition}
%
Observe that the expectation with respect to the distribution $\rho(s,a)$ is an integral of an infinite sum over a continuous space. In addition,  the system transition density $p(s_{t+1}|s_t,a_t)$ is not known. Therefore, computing \eqref{eqn_nabla_U} in closed form is intractable and a large number of samples might be needed to obtain an accurate Monte Carlo approximation even if $(p_{t+1}|s_t,a_t)$ was known. An alternative to overcome this drawback is the use of stochastic approximation methods (see \cite{robbins1951stochastic,kivinen2004online,zhang2004solving,pontil2005error}), where the main idea is to compute an unbiased estimate of the policy gradient by evaluating the expression inside the expectation for one sample of a pair $(s,a) \sim \rho(s,a)$, thus reducing the cost of each iteration. Observe however, that in this particular case the evaluation of the stochastic gradient requires the $Q$-function defined in \eqref{eqn_q_function}, which presents the same challenges that computing the gradient of the expected discounted reward, i.e., an intractable closed-form expression and a computationally prohibitive approximation. In Section \ref{sec_q_estimate} we present an efficient subroutine to find an unbiased estimate of the $Q$ function which is used in Section \ref{sec_nabla_U} to define the stochastic gradient of the expected discounted reward. 
In Section \ref{sec_convergence}, we show that by updating the policy with the stochastic estimate of $\nabla_h U(h,\cdot)$, convergence to a stationary point of $U(h)$ is achieved with probability one.  


\section{Stochastic Policy Gradient}\label{sec_algorithm}
In order to  compute a stochastic approximation of $\nabla_h U(h)$ we need to sample from $\rho(s,a)$ given in \eqref{eqn_discounted_distribution}. The intuition behind $\rho(s,a)$ is that it weights the probability of the system being at a specific state-action pair $(s,a)$ at time $t$ by a factor of $(1-\gamma)\gamma^t$. Notice that this factor is equal to the probability of a geometric distribution of parameter $\gamma$ to take the value $t$. Thus, for the $k$-th policy update, one can interpret the distribution $\rho(s,a)$ as the probability  of running the system for $T$ steps, with $T$ randomly drawn from a geometric distribution of parameter $\gamma$. This supports steps 2-7 in Algorithm \ref{alg_stochastic_grad} which describes how to obtain a sample $(s_k,a_k)\sim\rho(s,a)$. Later, in Proposition \ref{prop_unbiased_grad}  it is claimed that an unbiased estimate of $\nabla_h U(h)$ is obtained by substituting the sample $(s_k,a_k)$  in  
\begin{equation}\label{eqn_first_stochastic_gradient}
\hat{\nabla}_h U(h,\cdot) = \frac{1}{1-\gamma}\hat{Q}(s_k,a_k;h)\kappa(s_k,\cdot)\Sigma^{-1}(a_k-h(s_k)),
\end{equation}
with $\hat Q(s_k,a_k;h)$ being an unbiased estimate of $Q(s_k,a_k;h)$. The previous expression reveals a second challenge in  computing of the stochastic gradient, namely the need of computing the function $Q$ -- or an estimate-- at the state-action pair $(s_k,a_k)$. We deal with this  in Section \ref{sec_q_estimate}, providing an unbiased estimate of $Q(s_k,a_k;h)$ that yields an unbiased estimate of $\nabla_h U(h,\cdot)$ when substituted in \eqref{eqn_first_stochastic_gradient}. This unbiased estimate is constructed in a finite number of steps. Using this estimate and a step size $\eta_k>0$ we propose to update the policy iteratively following the rule
\begin{equation}\label{eqn_stochastic_update}
  h_{k+1} = h_k + \eta_k \hat{\nabla}_h U(h_k,\cdot),
\end{equation}
 Under proper conditions stochastic gradient ascent methods can be shown to converge with probability one to the local maxima  \cite{pemantle1990nonconvergence}. This approach has been widely used to solve parametric optimization problems where the decision variables are vectors in $\mathbb{R}^n$. In this paper we extend these results to non-parametric problems in RKHSs. First, we describe the algorithm to obtain the unbiased estimate $\hat Q(s_k,a_k;h)$ in a finite number of steps, which is instrumental for our overall non-parametric stochastic approximation strategy.
%
\subsection{Unbiased Estimate of Q}\label{sec_q_estimate}
A theoretically conceivable but unrealizable form of estimating the value of $Q(s,a;h)$ is to run a trajectory for infinite steps starting from $(s_0,a_0)=(s,a)$ and then compute
%
$\hat{q}_h = \sum_{t=0}^\infty \gamma^{t} r(s_t,a_t).$  
%
%
%
Despite being unbiased, the previous estimate requires an infinite number of steps. In contrast, we present the subroutine  Algorithm \ref{alg_q_estimate} that allows to compute an unbiased estimate of $Q(s,a;h)$ by considering a representative future reward obtained after a finite number of steps. 
As with $U(h)$, a parameter $\gamma$  closer to one  assigns similar weights to present and future rewards, and $\gamma$ close to zero prioritizes the present.   
%
%
In that sense, when $\gamma$ is very small, we do not need to let the system evolve for long time to get a representative reward. 
Again, the geometric distribution allows us to represent this idea. Specifically, let $T_Q$ be a geometric random variable with parameter $\gamma$, i.e., $P(T_Q=t) =(1-\gamma)\gamma^t$, which is finite with probability one. 
Then define the estimate of $Q(s,a;h)$ as the sum of rewards collected from step $t=0$ until $t=T_Q$ 
\begin{align}
     \hat{Q}(s,a;h) &:= (1-\gamma)\sum_{t=0}^{T_Q}r(s_t,a_t) \label{eqn_q_estimate} 
 \end{align}
Algorithm \ref{alg_q_estimate} summarizes how to obtain $\hat{Q}(s,a;h)$ as in \eqref{eqn_q_estimate}, and Proposition 1 states that it is unbiased.
%
\begin{algorithm}[t]
  \caption{estimateQ}
  \label{alg_q_estimate} 
\begin{algorithmic}[1]
 \renewcommand{\algorithmicrequire}{\textbf{Input:}}
 \renewcommand{\algorithmicensure}{\textbf{Output:}}
 \Require $s,a,h$
 \State \textit{Initialize}: $\hat{Q} = 0$, $s_0 = s$, $a_0 =a$
  \State Draw an integer $T_Q$ form a geometric distribution with parameter $\gamma$, $P(T_Q = t) = (1-\gamma)\gamma^t$
  \For {$t = 0,1,\ldots T_Q-1$}
  \State Collect reward and add to estimate $\hat{Q} = \hat{Q} + r(s_t,a_t)$
  \State Let system advance $s_{t+1} \sim P_{s_t\to s_{t+1}}^{a_t}$
  \State Select action $a_{t+1} \sim \pi_h(a|{s_{t+1}})$ 
  \EndFor
  \State Collect last reward $\hat{Q} = \hat{Q} + r(s_{T_Q},a_{T_Q})$
  \State Scale $\hat{Q} = (1-\gamma)\hat{Q}$\\	
 \Return $\hat{Q}$, $s_{T_Q}$ 
 \end{algorithmic}
 \end{algorithm}
%
\begin{proposition}\label{prop_unbiased_q}
  The output $\hat{Q}(s,a;h)$ of Algorithm \ref{alg_q_estimate} is an unbiased estimate of $Q(s,a;h)$. \end{proposition}
\begin{proof}
We start by writing the estimate $\hat{Q}(s,a;h)$ as 
    \begin{equation}\label{eqn_prop_unbiased_q_1}
    \begin{split}
 \hat{Q}(s,a;h)=(1-\gamma)\sum_{t=0}^\infty \mathbbm{1}(T_Q \geq t) r(s_t,a_t) ,\end{split}
    \end{equation}
  where we substituted $\infty$ for the $T_Q$ as the last index of the sum, but added null summands  for $t>T_Q$ by using the indicator function $\mathbbm{1}$. To show that it is unbiased we compute its expectation conditioning on $h$ and the initial state--action pair.
  %
  %
Notice that according to Algorithm \ref{alg_q_estimate} $T_Q$ is drawn independently of the system evolution. Furthermore, it will be argued below that the sum and expectation can be exchanged, with this in mind we write the expectation of \eqref{eqn_prop_unbiased_q_1} as 
\begin{align}
  &\mathbb{E}\left[ \hat{Q}(s,a;h)\Big| h,s_0=s, a_0=a\right]\label{q-sum-exchanged} \\
  &= (1-\gamma) \sum_{t=0}^\infty \mathbb{E}\left[\mathbbm{1}(T_Q \geq t)\right] \mathbb{E}\left[r(s_t,a_t) \Big| h,s_0=s, a_0=a\right].\nonumber
  \end{align}
Because $T_Q\sim$Geom$(\gamma)$ we have that $\mathbb{E}\left[\mathbbm{1}(T_Q \geq t)\right]=\gamma^t$ and
\begin{align}
  &\mathbb{E}\left[ \hat{Q}(s,a;h)\Big| h,s_0=s, a_0=a\right]\label{q-sum-exchanged-gamma} \\
	&=(1-\gamma) \sum_{t=0}^\infty \gamma^t \mathbb{E}\left[r(s_t,a_t) \Big| h,s_0=s, a_0=a\right]={Q}(s,a;h)\nonumber
  \end{align}
	%
		It remains to prove that the sum and the expectation in the previous expression are exchangeable. Using Assumption \ref{assumption_reward_function} and the triangle inequality, for all $N>0$ we have that 
  \begin{equation}
\left|\sum_{t=0}^N \mathbbm{1}(T_Q \geq t) r(s_t,a_t)\right| \leq \sum_{t=0}^N \mathbbm{1}(T_Q \geq t) B_r.
 \end{equation}
  Which by virtue of the monotonicity and the linearity of the expectation implies that 
  \begin{equation}\label{eqn_bound_auxiliary_expectations}
\mathbb{E}\left[\left|\sum_{t=0}^N \mathbbm{1}(T_Q \geq t) r(s_t,a_t)\right|\right]\leq B_r\mathbb{E}\left[\sum_{t=0}^N \mathbbm{1}(T_Q \geq t) \right].
  \end{equation}
 Observe that the random variable on the right is a monotonic increasing random variable and thus, by virtue of the monotone convergence theorem we have that
 \begin{align}\label{eqn_sum_of_indicators}
     \mathbb{E}\left[\sum_{t=0}^\infty \mathbbm{1}(T_Q \geq t) \right] =  \sum_{t=0}^\infty  \mathbb{E}\left[ \mathbbm{1}(T_Q \geq t) \right]  
     =\sum_{t=0}^\infty \gamma^t      = \frac{1}{1-\gamma}.
    \end{align}
Notice that the sequence $\left|\sum_{t=0}^N \mathbbm{1}(T_Q \geq t) r(s_t,a_t)\right|$ is dominated by $\sum_{t=0}^\infty \mathbbm{1}\left(T_Q\geq t\right) B_r$ for all $N\geq 0$ and that the latter has bounded expectation. Then, the Dominated Convergence Theorem applies (see e.g., \cite[Theorem 1.6.7]{durrett2010probability}), and guarantees that expectation and  sum can be exchanged in \eqref{eqn_prop_unbiased_q_1}.
  \end{proof}
%
\subsection{Unbiased Estimate of the Stochastic Gradient}\label{sec_nabla_U}
In this section we present a subroutine that uses the estimate  $\hat Q(s,a;h)$ produced by Algorithm  \ref{alg_q_estimate}   to obtain an unbiased estimate of $\nabla_hU(h)$. As discussed before, a sample from $\rho(s,a)$ can be obtained by sampling a time $T$ from a geometric  distribution of parameter $\gamma$ and running the system $T$ times. Although the resulting estimate in \eqref{eqn_first_stochastic_gradient} can be shown to be unbiased, which would be enough for the purpose of stochastic approximation, we chose to introduce  symmetry with respect to $h(s)$ as it is  justified in Section \ref{sec_approximation}. Instead of computing the approximation only at the state-action pair $(s_T,a_T)$ we average said value with  $\hat Q(s_T,\bar{a}_T)$, where $\bar{a}_T=h(s_T)-(a_T-h(s_T))$ is the symmetric action to $a_T$ with respect to $h(s_T)$ (steps 8--11 in Algorithm \ref{alg_stochastic_grad}). Hence, the proposed estimate is
\begin{align}\nonumber  
    \hat{\nabla}_hU(h,\cdot) &=\frac{1}{2(1-\gamma)}\left(\hat{Q}(s_T,a_T;h)-\hat{Q}(s_T,\bar{a}_T;h)\right)\\&\times \kappa(s_T,\cdot)\Sigma^{-1}(a_T-h(s_T))\label{eqn_stochastic_grad}.
    \end{align}
The subroutine presented in Algorithm \ref{alg_stochastic_grad} summarizes the algorithm to compute our stochastic approximation in \eqref{eqn_stochastic_grad}. We claim that it is unbiased in the following proposition. 
%
\begin{algorithm}[t]
 \caption{StochasticGradient}\label{alg_stochastic_grad}
\begin{algorithmic}[1]
 \renewcommand{\algorithmicrequire}{\textbf{Input:}}
 \renewcommand{\algorithmicensure}{\textbf{Output:}}
 \Require $h$, $s$
 \State \textit{Initialize}:  $s_0=s$ 
  \State Draw an integer $T$ from a geometric distribution with parameter $\gamma$, $P(T = t) = (1-\gamma)\gamma^t$
  \State Select action $a_{0} \sim \pi_h(a|{s_0})$ 
  \For {$t = 0,1,\ldots T-1 $}
  \State Advance system $s_{t+1} \sim P_{s_t\to s_{t+1}}^{a_t}$
  \State Select action $a_{t+1} \sim \pi_h(a|{s_{t+1}})$ 
  \EndFor
  \State \label{stepQ} Get estimate of $Q(s_T,a_T;h)$ as in Algorithm \ref{alg_q_estimate}: $$\hat{Q}(s_T,a_T;h) = \textrm{estimateQ}(s_T,a_T;h)$$  \vspace{-3mm}
  \State  Given $a_T$, find symmetric  $\bar{a}_T=h(s_T)-(a_T-h(s_T))$
  \State Get estimate of  $Q(s_T,\bar{a}_T;h)$ as in Algorithm \ref{alg_q_estimate}: $$\hat{Q}(s_T,\bar{a}_T;h)= \textrm{estimateQ}(s_T,\bar{a}_T;h)$$ \label{stepQ2}\vspace{-4mm} 
  \State  Compute the stochastic gradient $\hat{\nabla}_hU(h,\cdot)$ as in \eqref{eqn_stochastic_grad}
 \Return $\hat{\nabla}_hU(h,\cdot)$  
 \end{algorithmic}\label{alg_stochastic_gradient} 
\end{algorithm}
%
\begin{proposition}\label{prop_unbiased_grad}
  The output $\hat{\nabla}_h U(h,\cdot)$ of Algorithm \ref{alg_stochastic_gradient} is an unbiased estimate of $\nabla_h U(h,\cdot)$ in \eqref{eqn_nabla_U}.
\end{proposition}
\begin{proof}
  To show that the estimate is unbiased we write the expectation of $\hat{\nabla}_hU(h,\cdot)$ conditioned to $h$ as
  \begin{equation}
    \begin{split}
      \mathbb{E}\left[\hat{\nabla}_hU(h,\cdot)\Big| h\right] &= \mathbb{E}\left[\left[\hat{\nabla}_hU(h,\cdot)\Big|s_T,a_T\right]\Big| h\right] 
\end{split}
    \end{equation}
Using the linearity of the expectation and the fact that $\kappa(s_T,\cdot)\Sigma^{-1}(a_T-h(s_T))$ is measurable with respect of the sigma algebra generated by $(s_0,a_0 \ldots s_T,a_T)$ we have that 
\begin{equation}
   \begin{split}
     \mathbb{E}\left[\hat{\nabla}_hU(h,\cdot)\Big| h\right] &=
      \mathbb{E}\left[\mathbb{E}\left[\hat{Q}(s_T,a_T;h)-\hat{Q}(s_T,\bar{a}_T;h)\Big|s_T,a_T\right]\right.\\
&        \left.\frac{1}{2(1-\gamma)}\kappa(s_T,\cdot)\Sigma^{-1}(a_T-h(s_T))\Big| h\right].
   \end{split}
\end{equation}
By virtue of Proposition \ref{prop_unbiased_q} the previous expression reduces to
\begin{equation}
   \begin{split}
     \mathbb{E}\left[\hat{\nabla}_hU(h,\cdot)\Big| h\right] = 
      \mathbb{E}\left[\left({Q}(s_T,a_T;h)-{Q}(s_T,\bar{a}_T;h)\right)\right.\\
        \left.\frac{1}{2(1-\gamma)}\kappa(s_T,\cdot)\Sigma^{-1}(a_T-h(s_T))\Big| h\right].
   \end{split}
\end{equation}
%
Since $a_T$ is normally distributed with mean $h(s_T)$ we have that $\eta_T:=a_T-h(s_T)$ and $h(s_T)-a_T$ are both normally distributed with zero mean. Moreover, $\bar{a}_T$ has the same distribution as $a_T$. Hence the two terms on the right hand side of the previous equality are the same. Adding them yields
\begin{equation}\label{eqn_before_swaping_sum_expectation}
  \begin{split}
    \mathbb{E}\left[\hat{\nabla}_hU(h,\cdot)\Big| h\right]= \mathbb{E}\left[\frac{Q(s_T,a_T;h)}{1-\gamma} \kappa(s_T,\cdot)\eta_T\Big|h\right] \\
=    \frac{1}{1-\gamma}\mathbb{E}\left[\sum_{t=0}^\infty \mathbbm{1}(T=t)Q(s_t,a_t;h)\kappa(s_t,\cdot)\eta_t|h\right].
    \end{split}
\end{equation}
Next, we argue that it is possible to exchange the infinity sum and the expectation in the previous expression. Observe that only one of terms inside the sum can be different than zero. Denote by $t^*$ the index corresponding to that term and upper bound the norm of 
%
%
%
$\hat{\nabla}_h U(h)$ by
\begin{equation}
  \begin{split}
    \hspace{-0.5cm}(1-\gamma)\left\|\hat{\nabla}_h U(h)\right\|\leq \left|Q(s_{t^*},a_{t^*};h)\right|\left\| \kappa(s_{t^*},\cdot)\right\|\left\|\eta_{t^*}\right\|.
    \end{split}
\end{equation}
Using that $\left\|\kappa(s_t,\cdot)\right\| =1$ (cf., Definition \ref{def_rkhs}) and  $\left|Q(s,a;h) \right| \leq B_r/(1-\gamma)$ (cf., Lemma \ref{lemma_bounded} in the Appendix), we can upper bound the previous expression by  
%
%
\begin{equation}\label{eqn_bound_kappa_Q}
  \begin{split}
    \left\|\hat{\nabla}_h U(h)\right\|&\leq \frac{B_r}{\left(1-\gamma\right)^2}\left\|\eta_{t^*}\right\|  
  \end{split}
\end{equation}
Notice that $\Sigma^{1/2}\eta_{t} =\Sigma^{-1/2}\left(a_t-h(s_t)\right)$ are identically distributed mutlivariate normal distributions, and thus the expectation of its norm is bounded. The Dominated Convergence Theorem can be hence used to exchange the sum and the expectation in \eqref{eqn_before_swaping_sum_expectation}. In addition, the draw of the random variable $T$ is independent of the evolution of the system until infinity. Hence \eqref{eqn_before_swaping_sum_expectation} yields
\begin{align}
      \mathbb{E}\left[\hat{\nabla}_hU(h,\cdot)\Big| h\right]=     \sum_{t=0}^\infty \frac{P\left(t=T\right)}{1-\gamma}\mathbb{E}\left[Q(s_t,a_t;h)\kappa(s_t,\cdot)\eta_t|h\right]  \nonumber\\
 =\sum_{t=0}^\infty \gamma^t\mathbb{E}\left[Q(s_t,a_t;h)\kappa(s_t,\cdot)\eta_t)|h\right]=\nabla_hU(h,\cdot), \label{eqn_after_swaping_sum_expectation}
  \end{align}
where the last equality coincides with that in \eqref{eqn_nabla_U}. To be able to write the last equality we need to justify that sum and expectation are exchangeable, which we do next. Let us define the following sequence of random variables
\begin{equation}
S_k = \sum_{t=0}^k\gamma^t Q(s_t,a_t;h)\kappa(s_t,\cdot)\Sigma^{-1}(a_t-h(s_t)).
  \end{equation}
Use the triangle inequality along with the bounds for $\kappa(s_t,\cdot)$ and $Q(s_t,a_t;h)$ from \eqref{eqn_bound_kappa_Q} to bound the norm of $S_k$ by
\begin{equation}
\left\|S_k\right\| \leq  \frac{B_r}{1-\gamma}\sum_{t=0}^k\gamma^t\left\|\Sigma^{-1}(a_t-h(s_t))\right\|.
\end{equation}
Observe that the sum in the right is an increasing random variable because all terms in the sumands are positive. Hence, by virtue of the Monotone Convergence Theorem (see e.g., \cite[Theorem 1.6.6]{durrett2010probability}) we have that 
\begin{equation}
\mathbb{E}\hspace{-.1cm}\left[\sum_{t=0}^\infty\hspace{-.1cm}\gamma^t\left\|\Sigma^{-1}(a_t-h(s_t))\right\|\right]=\sum_{t=0}^\infty\hspace{-.1cm}\gamma^t\mathbb{E}\hspace{-.1cm}\left[\left\|\Sigma^{-1}(a_t-h(s_t))\right\|\right].
\end{equation}
Because $\Sigma^{-1/2}(\left(a_t-h(s_t)\right)$ is normally distributed, its norm has bounded expectation. Use in addition the fact that the geometric series converges to ensure that the right hand side of the previous expression is bounded. $S_k$ is therefore dominated by a random variable with finite expectation. Thus, the Dominated Convergence Theorem  allows us to write that
\begin{equation}
\lim_{k\to \infty} \mathbb{E}[S_k] = \mathbb{E}[\lim_{k\to\infty}S_k],
\end{equation}
which implies that \eqref{eqn_after_swaping_sum_expectation} holds.
\end{proof}
Now we are in conditions of presenting the complete algorithm for policy gradient in RKHSs. Each iteration consists of the estimation of $\hat{\nabla}_h U(h_k,\cdot)$ as described in Algorithm \ref{alg_stochastic_gradient} -- which uses Algorithm  \ref{alg_q_estimate} as a subroutine to get unbiased estimates of $Q(s,a;h)$ -- and of the updated
\begin{equation}\label{eqn_gradient_ascent}
{h}_{k+1} = h_k +\eta_k\hat{\nabla}_h U(h_k,\cdot),
  \end{equation}
where $\eta_k$ is non-summable and square summable, i.e.
\begin{equation}\label{eqn_stepsize}
  \sum_{k=0}^\infty \eta_k = \infty \quad \mbox{and} \quad \sum_{k=0}^\infty \eta_k^2 <\infty.
  \end{equation}
%
\begin{algorithm}[t]
 \caption{Stochastic Policy Gradient Ascent}
\begin{algorithmic}[1]
 \renewcommand{\algorithmicrequire}{\textbf{Input:}}
 \renewcommand{\algorithmicensure}{\textbf{Output:}}
 \Require step size $\eta_0$
 \State \textit{Initialize}: $h_0=0$
 \For{$k=0 \ldots$}
\State Draw an initial state $s_0$ for Algorithm \ref{alg_stochastic_gradient}
 \State Compute the stochastic gradient: $$\hat{\nabla}_h U(h_k,\cdot) = \textrm{StochasticGradient}(h_k,s_0)$$
 \State Gradient ascent step ${h}_{k+1} = h_k +\eta_k \hat{\nabla}_h U(h_k,\cdot)
 $
 \EndFor
 \end{algorithmic}\label{alg_policy_gradient}
 \end{algorithm}
%
\begin{theorem}\label{thm_convergence_unbiased}
  Let $\left\{h_k,k\geq 0\right\}$ be the sequence of functions given by \eqref{eqn_gradient_ascent}, where $\eta_k$ is as step size satisfying \eqref{eqn_stepsize} and $\hat{\nabla}_hU(h_k,\cdot)$ is an unbiased estimator of the gradient of the functional. With probability one we have that $\lim_{k\to \infty} h_k = H^*$, 
  %
  where $H^*$ is a random variable taking values in the set of critical points of the functional $U(h)$ defined in \eqref{eqn_problem_statement}.
\end{theorem}
\begin{proof}
  The proof of this result is the matter of Section \ref{sec_convergence}.
\end{proof}
%
%
The previous result establishes that  $h_k$  converges with probability one to a critical point of the functional $U(h)$. A major drawback of Algorithm 3 is that at each iteration the stochastic gradient ascent iteration will add a new element to the kernel dictionary. Indeed, for each iteration $\hat{\nabla}_hU(h_k,\cdot)$ introduces an extra  kernel  centered at a new  $s_{T}$ (cf., \eqref{eqn_stochastic_grad}). Hence for any $k>0$ in order to represent $h_k$ we require $k$ dictionary elements. This translates into  memory explosion and thus Algorithm 3, while theoretically interesting, is not practical. To overcome this limitation, we introduce in the next section a projection on a smaller Hilbert space so that we can control the model order. Before that, we 
introduce a discussion regarding the use of random policies. . 
\subsection{Gaussian policy as an approximation} \label{sec_approximation}
Our reason to use a randomized Gaussian policy is two-fold: it yields a good approximation of the gradient of the Q-function that would result from a deterministic policy as we show in Proposition \ref{prop_approximation}, and it effects numerical derivatives when the gradients are handled via stochastic approximation (see also \cite{nesterov2017random}).  Building on these hints, we will propose alternative estimates for faster convergence. In this direction, we consider the Gaussian bell $\pi_h(a|{s})$ with covariance $\Sigma$ as an approximation to the Dirac's impulse  \cite{schwartz1998distributions}, and its gradient $\nabla_a \pi_h(a|{s})= \Sigma^{-1} (a-h(s))\pi_h(a|{s})$ as an approximation of the impulse's gradient. Then, the next proposition follows

\begin{proposition}\label{prop_approximation}
  Consider a family of Gaussian policies with  $\Sigma$ 
  and let $U_\Sigma(s;h)$ and $Q_\Sigma(s,a;h)$ be the cumulative rewards and Q-functions that results from such policies, respectively. Correspondingly, let $Q_0(s,a;h):=\mathbb{E}\left[\sum_{t=0}^\infty \gamma^t r(s_t,a_t)\Big| h, s_0=s, a_0=a\right]$ be the Q-function that results from  a deterministic policy $a_t=h(s_t)$. 
Let  $\nabla_a Q_\Sigma(s,a;h)$ be bounded for all $s,a,h$ and $\Sigma$, then  
\begin{align}
  \lim_{\Sigma\to 0}\int Q_\Sigma(s,a;h) \nabla_a\pi_h(a|{s})  da=\nabla_a Q_0(s,a;h)\label{pi_is_gradient}
\end{align}
and defining $\rho(s)$ such that $\rho(s,a)=\rho(s)\pi_h(a|{s}).$   
\begin{align*}
\lim_{\Sigma \to 0}\nabla_h & U_\Sigma(h,\cdot)  &= \frac{1}{1-\gamma} \int \nabla_a Q_0(s,a;h)\rho(s) \kappa(s,\cdot) ds.\end{align*}
\end{proposition}
\begin{proof}
  Integrating by parts the expression \eqref{pi_is_gradient} yields 
    \begin{equation}
      \begin{split}
        \int Q_\Sigma(s,a;h) \nabla_a\pi_h(a|{s})  da \\
        = -Q_{\Sigma}(s,a;h)\pi_h(a|{s})\Big|_{-\infty}^\infty +  \int \nabla_a Q_\Sigma(s,a;h) \pi_h(a|s) \, da.
\end{split}
      \end{equation}
    The first term is zero because $Q_{\Sigma}(s,a,h)$ is bounded for all $s,a,h$ and $\Sigma$ (cf., Lemma \ref{lemma_bounded}) and the Gaussian goes to zero at infinity. To work with the second integral, consider the following variable $\eta = \Sigma^{-1/2}\left(a-h(s)\right)$. By introducing this change of variable $\pi_h(a|s)da =  \phi(\eta) d\eta$, where $\phi(\eta)$ is the multivariate normal distribution. Hence, it follows that 
    \begin{equation}
      \begin{split}
            \int Q_\Sigma(s,a;h) \nabla_a\pi_h(a|{s})  da
            = \int \nabla_a Q_\Sigma(s,a;h) \pi_h(a|s) \, da \\
            =\int \nabla_a Q_\Sigma(s,\Sigma^{1/2}\eta+h(s),h) \phi(\eta)\,d\eta.
            \end{split}
\end{equation}
Because $\nabla_a Q_\Sigma(s,\Sigma^{1/2}\eta+h(s),h)$ is bounded for all $s,a,h$ and $\Sigma$ we can use the Dominated Convergence Theorem to exchange the limit and the integral in \eqref{pi_is_gradient}. Hence, 
\begin{equation}
  \begin{split}
    \lim_{\Sigma\to 0}\int Q_\Sigma(s,a;h) \nabla_a\pi_h(a|{s})  da \\
    =\int\lim_{\Sigma\to 0} \nabla_a  Q_\Sigma(s,\Sigma^{1/2}\eta+h(s);h) \phi(\eta) \, d\eta.
  \end{split}
\end{equation}
We will show afterwards that $\lim_{\Sigma\to 0} Q_\Sigma (s,a;h) = Q_0(s,a;h)$ the Q-function for a deterministic policy $a_t = h(s_t)$. This being the case the previous integral reduces to 
\begin{equation}
  \begin{split}
    \lim_{\Sigma\to 0}\int Q_\Sigma(s,a;h) \nabla_a\pi_h(a|{s})  da    & =\int \nabla_a  Q_0(s,h(s);h) \phi(\eta) \, d\eta \\&= \nabla_a  Q_0(s,h(s);h),
  \end{split}
\end{equation}
where in the previous expression we had swaped the derivative with respect to $a$ and the limit. The proof of this is analogous to the proof that $\lim_{\Sigma\to 0} Q_\Sigma (s,a;h) = Q_0(s,a;)$ the Q-function that results from a deterministic policy $a_t = h(s_t)$. We do this next to complete the proof. Observe that for any $\Sigma$ the Q-function can be written as
  \begin{equation}
    \begin{split}
    Q_\Sigma(a_0,s_0;h) = \\
    \sum_{t=0}^\infty\gamma^t \int r(s_t,a_t) \prod_{u=0}^{t-1}p(s_{u+1}|s_u,a_u) \pi_h(a_{u+1},s_{u+1}) \, d\bbs d\bba.
    \end{split}
    \end{equation}
  Taking the limit with $\Sigma \to 0$, we have that $\pi_h(a|s) = \delta(a-h(s))$. Hence, the previous expression yields
  \begin{equation}
    \begin{split}
      \lim_{\Sigma \to 0} Q_\Sigma(a_0,s_0;h) =r(s_0,a_0)+  \\
      \sum_{t=1}^\infty \gamma^t\int r(s_t,a_t) \prod_{u=0}^{t-1}p(s_{u+1}|s_u,a_u) \delta(a_{u+1}-h(s_{u+1}) \, d\bbs d\bba. \\
      = r(s_0,a_0) \\
      +\sum_{t=1}^\infty \gamma^t\int r(s_t,h(s_t)) p(s_1|s_0,a_0)\prod_{u=1}^{t-1}p(s_{u+1}|s_u,h(s_u)) \, d\bbs.
      \end{split}
    \end{equation}
  Which shows that that $\lim_{\Sigma \to 0} Q_\Sigma(s,a;h)$ is indeed the Q-function for a deterministic policy $a_t = h(s_t)$. The proof of the second part of the proposition follows analogously. 
\end{proof}
The assumption of $\nabla_a Q_\Sigma(s,a;h)$ being bounded is satisfied if for instance the derivatives of $r(s,a)$ and $p(s_{t+1}|s,a)$ with respect to $a$ are bounded. 
 This interpretation of the integral in \eqref{pi_is_gradient} as the gradient of $Q(s,a;h)$ can be seen from the perspective of stochastic approximation. For notational brevity  we define $I_\pi:=\int Q(s,a;h) \nabla_a \pi_h(a|{s})  da$,  and express it in terms of expectations
\begin{align}I_\pi&=\mathbb E_{a\sim\pi_{h}}[Q(s,a;h) \Sigma^{-1} (a-h(s))]\label{IpiQ}
\end{align}
Then, an unbiased stochastic approximation can be obtained 
%
by sampling two (or more) instances $a$ and $a'$ from     $\pi_h(a|{s})$ and averaging  as in $\hat I_\pi= \frac{1}{2} Q(s,a;h) \Sigma^{-1} (a-h(s))+ \frac{1}{2} Q(s,a';h) \Sigma^{-1} (a'-h(s))$. Furthermore, if  $a'$ is the symmetric action of $a$ with respect to the mean $h(s)$, then the estimator  is still unbiased.   Define the zero-mean Gaussian variable $\eta = a-h(s)$ to be the deviation of $a$ from $h(s)$. Thus by symmetry,    $a'-h(s)=-\eta$, and we can  rewrite the symmetric estimate as the finite difference
\begin{align}\hspace{-.1cm} \hat I_\pi&=\frac{\Sigma^{-1}\eta}{2}(Q(s,h(s)+\eta;h)-Q(s,h(s)-\eta;h)),  
\label{numerical_derivative}\end{align}
revealing the  gradient structure hidden in \eqref{IpiQ}. 
The interpretation of \eqref{numerical_derivative} as a derivative is relevant to our policy method  because it  reveals the underlying reinforcement mechanisms, in the sense that the policy update favors directions that improve the reward. 
Fig.\ref{fig_gradient} (left) represents the field $Q(s,a;h)$ as a function of $a\in\mathbb R^2$, and
the gradient estimate $\hat I_\pi$ in \eqref{numerical_derivative} that is obtained by weighting two opposite directions with the corresponding rewards. Since the reward in the direction $\eta$ is relatively higher $\hat I_\pi(Q)$ points in this direction. 
%
Fig. \ref{fig_gradient} (right) shows that the direction of $\nabla_a Q(s,a;h)$ can be approximated more accurately at the expense of sampling the reward at $2d$ points in quadrature. 
\begin{figure}
	\centering
		\includegraphics[width=0.49 \linewidth]{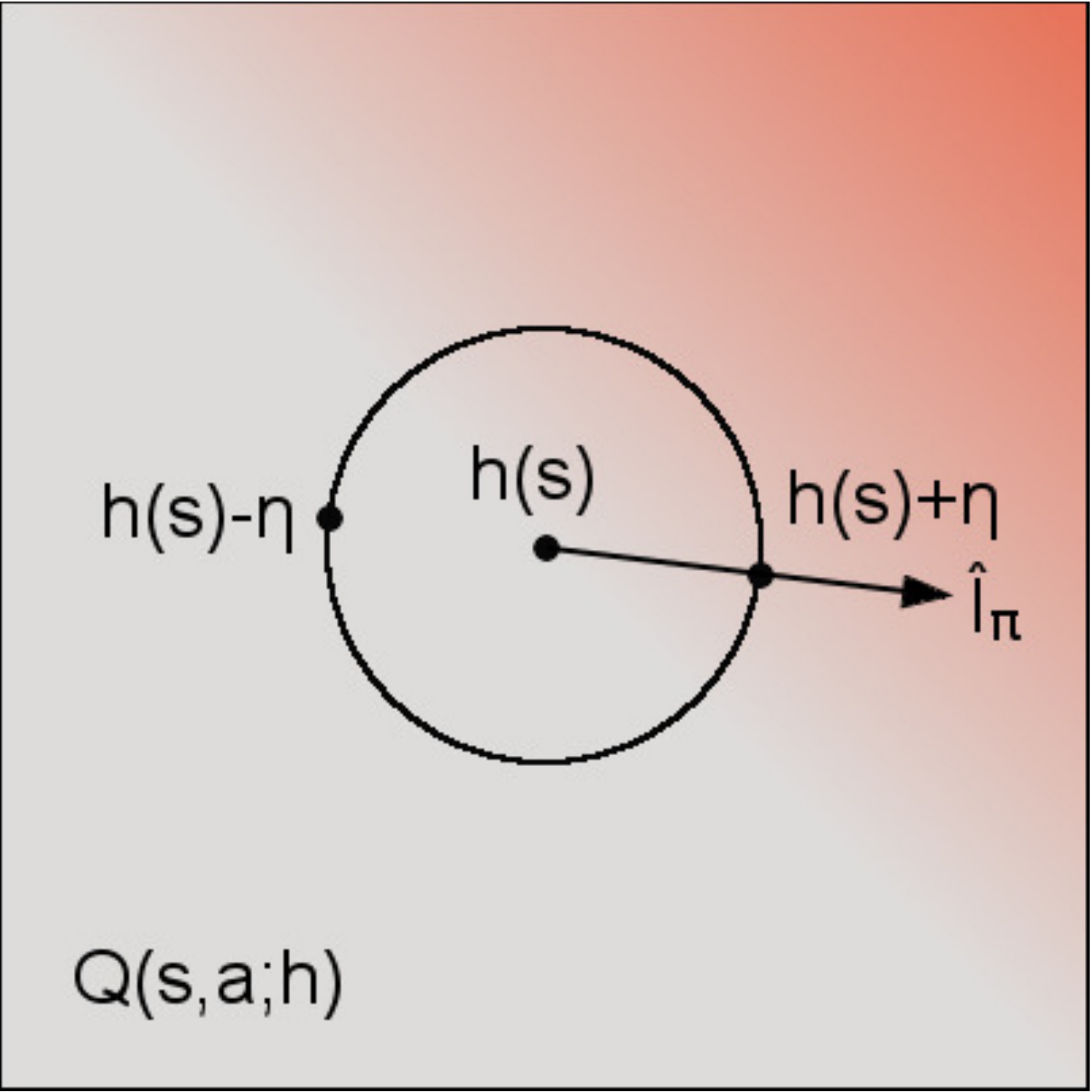} \includegraphics[width=0.49\linewidth]{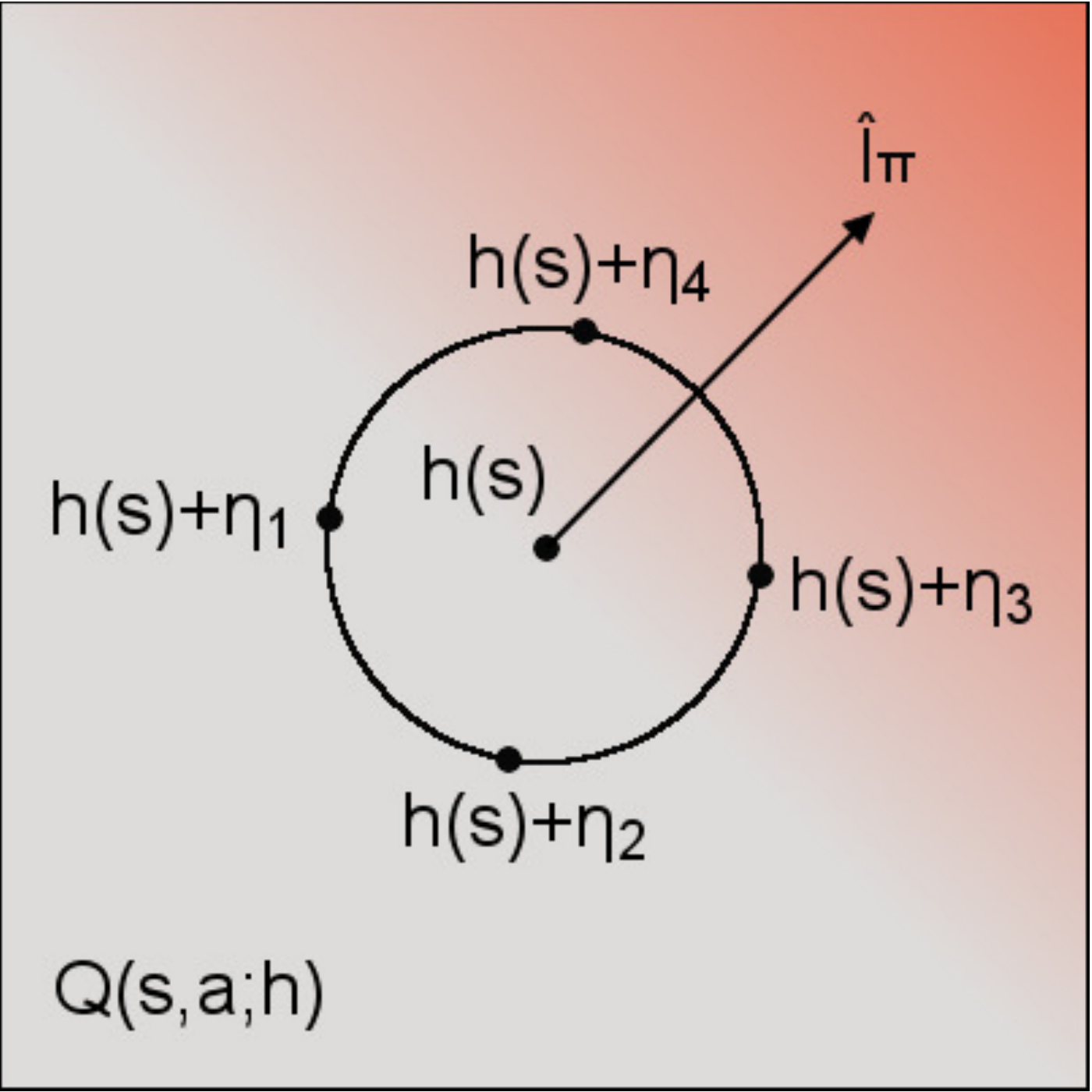}
	\caption{Numerical gradient via stochastic approximation; (left) two-sample approximation, (right) full-dimension. Red levels represents higher values of $Q(s,a;h)$.}
	\label{fig_gradient}
\end{figure}



\section{Convergence Analysis for Unbiased Stochastic Gradient Ascent}\label{sec_convergence}
This section contains the proof of Theorem \ref{thm_convergence_unbiased}. For this purpose let us introduce a probability space $\left(\Omega,\ccalF,P\right)$ and the following filtration defined as a sequence of increasing sigma-algebras $\left\{\emptyset, \Omega \right\} = \ccalF_0 \subset\ccalF_1 \subset \ldots \subset \ccalF_k \subset \ldots \subset \ccalF_{\infty}\subset \ccalF$, where $\ccalF_k$ is the sigma algebra generated by the random variables $h_0,\ldots,h_k$. Then, define the following constant $B=\left(L_1\sigma^2+L_2\eta_0\sigma^{3}\right)$, where $\sigma$ is the constant defined in Lemma \ref{lemma_stochastic_gradient} and $L_1$ and $L_2$ are those defined in Lemma \ref{lemma_lipchitz}. Next, consider the following sequence of random variables 
\begin{equation}\label{eqn_supermartingale}
  V_k = U(h_k)-B\sum_{j=k}^\infty \eta_j^2
  \end{equation}
Since the sequence $\eta_k$ is square summable and the expected discounted reward $U(h)$ is bounded (cf., Lemma \ref{lemma_bounded}), the random variable $V_k$ is bounded for all $k \geq 0$. We next show that the sequence \eqref{eqn_supermartingale} is a bounded submartingale. 
\begin{lemma}\label{lemma_supermartingale}
  The sequence $V_k$ defined in \eqref{eqn_supermartingale} is a bounded submartingale and it verifies that
\begin{equation}\label{eqn_supermartingale_equation}
    \begin{split}
    \mathbb{E}\left[{V_{k+1}|\ccalF_k}\right] \geq  V_k+\eta_k\left\|\nabla_h U(h_k)\right\|_\ccalH^2.
\end{split}
    \end{equation}
  \end{lemma}
\begin{proof}
   According to Lemma \ref{lemma_bounded} the value function $U(h_k)$ in \eqref{eqn_supermartingale} is upper-bounded. Thus $V_k$  is also upper bounded since  the stepsizes are square-summable according to \eqref{eqn_stepsize}. Observe as well that by definition $h_k\in\ccalF_k$ for all $k$ and therefore $V_k$ is adapted to the sequence of sigma-algebras. To show that $V_k$ is a submartingale it suffices to show \eqref{eqn_supermartingale_equation} which we do next. Writing the Taylor expansion of $U(h_{k+1})$ around $h_k$, yields
  \begin{equation}
    \begin{split}
     U(h_{k+1})
      &= U(h_k)+\left<\nabla_h U(f_k),h_{k+1}-h_k\right>_{\ccalH}, 
      \end{split}
  \end{equation}
 where $f_k = \lambda h_k +(1-\lambda)h_{k+1}$ with $\lambda\in[0,1]$. Adding and subtracting $\left<\nabla_h U(h_k),h_{k+1}-h_k\right>_{\ccalH}$ to the previous expression, using the Cauchy-Schwartz inequality and the result of Lemma \ref{lemma_lipchitz} we can lower bound $U(h_{k+1})$ by
  \begin{equation}
    \begin{split}
      &U(h_{k+1}) \geq U(h_k)+\left<\nabla_h U(h_k),h_{k+1}-h_k\right>_{\ccalH} \\
      &-L_1\left\|h_{k+1}-h_k\right\|_{\ccalH}^2-L_2\left\|h_{k+1}-h_k\right\|_{\ccalH}^3.
      \end{split}
  \end{equation}
  Let us consider the conditional expectation of the random variable $U(h_{k+1})$ with respect to the sigma-field $\ccalF_k$. Combine the monotonicity and the linearity of the expectation with the fact that $h_k$ is measurable with respect to $\ccalF_k$ to write
  \begin{equation}\label{eqn_first_supermartingale_bound}
    \begin{split}
      \mathbb{E}\left[{U(h_{k+1})|\ccalF_k}\right] \geq U(h_k)+\left<\nabla_h U(h_k),\mathbb{E}\left[h_{k+1}-h_k|\ccalF_k\right]\right>_{\ccalH} \\
      -L_1\mathbb{E}\left[\left\|h_{k+1}-h_k\right\|_{\ccalH}^2|\ccalF_k\right]
      -L_2\mathbb{E}\left[\left\|h_{k+1}-h_k\right\|_{\ccalH}^3|\ccalF_k\right].
      \end{split}
    \end{equation}
  Substitute $h_{k+1}$ by its expression in \eqref{eqn_stochastic_update} to write the expectation of the quadratic term as 
  \begin{equation}\label{eqn_bound_el1}
    \begin{split}
      L_1\mathbb{E}\left[\left\|h_{k+1}-h_k\right\|_{\ccalH}^2|\ccalF_k\right] &=  \eta_k^2 L_1\mathbb{E}\left[\left\|{\nabla}_h U(h_k,\cdot)\right\|_{\ccalH}^2|\ccalF_k\right].
      \end{split}
    \end{equation}
 Likewise, we have that 
  \begin{equation}\label{eqn_bound_el2}
    \begin{split}
      L_2\mathbb{E}\left[\left\|h_{k+1}-h_k\right\|_{\ccalH}^3|\ccalF_k\right] &\leq \eta_k^3 L_2\mathbb{E}\left[\left\|{\nabla}_h U(h_k,\cdot)\right\|_{\ccalH}^3|\ccalF_k\right] \\
      \end{split}
    \end{equation}
  Substituting \eqref{eqn_bound_el1} and \eqref{eqn_bound_el2} in \eqref{eqn_first_supermartingale_bound} and using the bounds for the moments of the stochastic gradient derived in Lemma \ref{lemma_stochastic_gradient} along with the fact that $\eta_k$ is nonincreasing and the definition of the constant $B=\eta_k^2L_1\sigma^2+\eta_k^2\eta_0L_2\sigma^3$, it results
  \begin{equation}
    \begin{split}
    \mathbb{E}\left[{V_{k+1}|\ccalF_k}\right] 
    \geq V_k+\left<\nabla_h U(h_k),\mathbb{E}\left[h_{k+1}-h_k|\ccalF_k\right]\right>_{\ccalH}.
    \end{split}
    \end{equation}
  To complete the proof observe that according to \eqref{eqn_stochastic_update} $h_{k+1}-h_k= \eta_k \hat \nabla_h U(h_k) $  and that the stochastic gradient is an unbiased estimate of the gradient (cf. Proposition \ref{prop_unbiased_grad}). 
  \end{proof}
%
The previous Lemma establishes that $V_k$ is a submartingale. A submartingale is in a sense a generalization of an increasing function and because it is bounded above it is expected that it converges. In fact this can be formalized (cf., \cite[Theorem 5.2.8]{durrett2010probability}). Moreover, the expression in \eqref{eqn_supermartingale_equation} and the convergence of $V_k$ suggest that the norm of the gradient $\left\| \nabla_h U(h_k,\cdot)\right\|$ goes to zero as $k$ goes to infinity. We show that this is the case in what follows.
  %
  %
  By virtue of Lemma \ref{lemma_supermartingale} it follows that the sequence $V_k$ defined in \eqref{eqn_supermartingale} is a bounded submartingale and therefore it converges almost everywhere to a limiting random variable $V:=\lim_{k\to\infty } V_k$ such that $\mathbb{E}|V|<\infty$ (cf., \cite[Theorem 5.2.8]{durrett2010probability}). Continuing the proof of Theorem \ref{thm_convergence_unbiased}, consider the conditional expectation of $V_{k+1}$ with respect to the sigma algebra $\ccalF_{k-1}$. Since $\ccalF_{k-1}\subset\ccalF_k$ it holds that
  \begin{equation}\label{eqn_conditional_exp_vk}
\mathbb{E}\left[V_{k+1}\big|\ccalF_{k-1}\right] = \mathbb{E}\left[\mathbb{E}\left[V_{k+1}\big|\ccalF_k\right]\big|\ccalF_{k-1}\right].
  \end{equation}
  Then, substitute  \eqref{eqn_supermartingale_equation} in \eqref{eqn_conditional_exp_vk} to obtain
  \begin{equation}
    \begin{split}
      \mathbb{E}\left[V_{k+1}\big|\ccalF_{k-1}\right] &\geq \mathbb{E}\left[V_k+\eta_k\left\|\nabla_h U(h_k,\cdot)\right\|^2\big|\ccalF_{k-1}\right] \\
      &= \mathbb{E}\left[V_k\big|\ccalF_{k-1}\right]+ \eta_k\mathbb{E}\left[\left\|\nabla_h U(h_k,\cdot)\right\|^2\big|\ccalF_{k-1}\right],
      \end{split}
  \end{equation}
		Next, use again \eqref{eqn_supermartingale_equation} to lower bound the first term on the right hand side of the previous equation
  \begin{equation}
    \begin{split}
      \mathbb{E}\left[V_{k+1}\big|\ccalF_{k-1}\right] &\geq  V_{k-1}+\eta_{k-1}\left\|\nabla_h U(h_{k-1},\cdot)\right\|^2 \\
      &+ \eta_k\mathbb{E}\left[\left\|\nabla_h U(h_k,\cdot)\right\|^2\big|\ccalF_{k-1}\right].
      \end{split}
    \end{equation}
    Repeating this procedure of conditioning on the previous sigma algebras recursively one obtains
    \begin{equation}
\mathbb{E}\left[V_{k+1}\right] \geq V_0+ \eta_0\left\|\nabla_h U(h_0,\cdot)\right\|^2 +\sum_{j=1}^{k}\eta_j\mathbb{E}\left[\left\|\nabla_h U(h_j,\cdot)\right\|^2\right].
    \end{equation}
    Since $V_k$ is a sequence of bounded random variables, by virtue of the Dominated Convergence Theorem we have that $\mathbb{E}\left[V \right] = \lim_{k\to\infty} \mathbb{E}\left[V_k\right]$. Hence, the previous inequality results in
    \begin{equation}
\mathbb{E}\left[V\right] \geq V_0 + \eta_0\left\|\nabla_h U(h_0,\cdot)\right\|^2 +\sum_{j=1}^{\infty}\eta_j\mathbb{E}\left[\left\|\nabla_h U(h_j,\cdot)\right\|^2\right],
      \end{equation}
with $\mathbb{E}|V|<\infty$, hence
    %
$\sum_{j=1}^{\infty}\eta_j\mathbb{E}\left[\left\|\nabla_h U(h_j,\cdot)\right\|^2\right] < \infty.$
    %
    The monotone convergence theorem applied to the sum $\sum_{j=1}^k \eta_j\left\| \nabla_h U(h_j,\cdot)\right\|^2 $ implies that
    \begin{equation}
      \lim_{k\to\infty} \mathbb{E}\left[\sum_{j=1}^k \eta_k\left\| \nabla_h U(h_j,\cdot)\right\|^2 \right] = \mathbb{E}\left[\sum_{j=1}^\infty \eta_k\left\| \nabla_h U(h_j,\cdot)\right\|^2 \right].
      \end{equation}
    Since the left hand side of the previous expression is bounded 
    %
    \begin{equation}\label{eqn_sum_gradient}
      \lim_{k\to\infty} \sum_{j=0}^k \eta_j\left\| \nabla_h U(h_j,\cdot)\right\|^2 <\infty \quad \mbox{a.e.}
    \end{equation}
 Because the sequence of step sizes $\eta_j$ is non-summable (cf., \eqref{eqn_stepsize}) the previous expression implies that
    \begin{equation}
\liminf_{k\to \infty} \left\|\nabla_h U(h_k,\cdot)\right\|^2 =0.
    \end{equation}

    %
    %
    %
    %
    We are left to show that $\limsup_{k\to \infty}\left\|\nabla_h U(h_k,\cdot)\right\|=0$ almost everywhere, which we do by contradiction. Assume that $\limsup_{k\to \infty}\left\|\nabla_h U(h_k(\omega),\cdot)\right\|= \epsilon>0$ for some $\omega\in\Omega$. Then, there exist subsequences $\{m_j\}$ and $\{n_j\}$ such that $m_j<n_j<m_{j+1}$ and that for $m_j\leq k <n_j $ we have that 
    \begin{equation}\label{eqn_higher_than_epsilon}
      \left\| \nabla_h U(h_k,\cdot)\right\|>\frac{\epsilon}{3}
\end{equation}
 and for $ n_j\leq k <m_{j+1}$ we have that 
    \begin{equation}\label{eqn_lower_than_epsilon}
      \left\| \nabla_h U(h_k,\cdot)\right\|\leq \frac{\epsilon}{3}, 
    \end{equation}
where we have dropped the $\omega$ to simplify the notation, but hereafter we argue for a specific sample point in the probability space. It is proved in Lemma \ref{lemma_square_int_martingale} in the appendix, that 
    \begin{equation}\label{eqn_square_integrable_martingale}
S_k =\sum_{j=0}^k \eta_j \left(\hat{\nabla}_h U(h_j)-{\nabla}_h U(h_j) \right) =\sum_{j=0}^k \eta_j e_j
     \end{equation}
    converges to a finite limit with probability one. By virtue of this result and \eqref{eqn_sum_gradient} there exists $\bar{j}$ such that
    \begin{equation}\label{eqn__bound_tails}
      \sum_{k=m_{\bar{j}}}^\infty \eta_k \left\|\nabla_h U(h_k,\cdot)\right\|^2 < \min\left\{\frac{\epsilon^2}{36L_1}, \frac{\epsilon^{3/2}}{6\sqrt{6L_2}}\right\} \end{equation}
    and
    \begin{equation}\label{eqn__second_bound_tails}
      \left\| \sum_{k=m_{\bar{j}}}^\infty \eta_ke_k\right\| < \min\left\{\frac{\epsilon}{12L_1}, \frac{\epsilon^{1/2}}{2\sqrt{6L_2}}\right\}
      \end{equation}
    For any $j\geq \bar{j}$ and any $m$ with $m_j\leq m <n_j$, by virtue of Lemma \ref{lemma_lipchitz}, we have
    \begin{equation}\label{eqn_difference_in_nablas}
      \begin{split}
        \left\|\nabla_h U(h_{n_j},\cdot)-\nabla_h U(h_m,\cdot)\right\|& \\
        \leq L_1\left\|h_{n_j}-h_m\right\|&+L_2\left\|h_{n_j}-h_m\right\|^2,
        \end{split}
    \end{equation}
    Recall that the difference $h_{n_j}-h_m$ can be written as
    \begin{equation}
      \begin{split}
        h_{n_j}-h_m&=\sum_{k=m}^{n_j-1}\eta_k \hat{\nabla}_h U(h_k,\cdot). 
    \end{split}
      \end{equation}
    Thus, defining the error $e_k =\hat{\nabla}_h U(h_k,\cdot)-{\nabla}_h U(h_k,\cdot)$, the following upper bound holds
    \begin{equation}\label{eqn_first_bound_diff_h}
      \begin{split}
        \left\|h_{n_j}-h_m\right\|&\leq \sum_{k=m}^{n_{j-1}}\eta_k\left\|\nabla_h U(h_k,\cdot)\right\| + \left\|\sum_{k=m}^{n_{j-1}}\eta_k e_k \right\| \\
        &\leq \frac{3}{\epsilon}\sum_{k=m}^{n_{j-1}}\eta_k\left\|\nabla_h U(h_k,\cdot)\right\|^2 + \left\|\sum_{k=m}^{n_{j-1}}\eta_k e_k \right\|,
        \end{split}
    \end{equation}
    where in the last inequality we used that that according to \eqref{eqn_higher_than_epsilon} for all $k$ such   $m\leq k <n_j$ we have that $(3/\epsilon)\left\| \nabla_h U(h_k,\cdot)\right\|\geq 1$. Using the bounds on the tails \eqref{eqn__bound_tails} and \eqref{eqn__second_bound_tails} it holds that
    \begin{equation}
\left\|h_{n_j}-h_m\right\|\leq \frac{3}{\epsilon}\frac{\epsilon^2}{36L_1} + \frac{\epsilon}{12L_1} = \frac{\epsilon}{6L_1}
      \end{equation}
        \begin{equation}
\left\|h_{n_j}-h_m\right\|\leq \frac{3}{\epsilon}\frac{\epsilon^{3/2}}{6\sqrt{6L_2}} + \frac{\epsilon^{1/2}}{2\sqrt{6L_2}} = \sqrt{\frac{\epsilon}{6L_2}}.
      \end{equation}
        Replacing the previous bounds in \eqref{eqn_difference_in_nablas} yields $        \left\|\nabla_h U(h_{n_j},\cdot)-\nabla_h U(h_m)\right\|\leq \epsilon/3$.
%
            The latter together with \eqref{eqn_lower_than_epsilon}  implies that $\left\|\nabla_h U(h_m,\cdot)\right\|<2\epsilon/3$ for all $m$ such   $m_j\leq m <n_j$, which contradicts \eqref{eqn_lower_than_epsilon} and therefore the assumption that $\limsup_{k\to\infty} \left\|\nabla_h U(h_k,\cdot)\right\| >0$. Hence, it must hold that $\lim_{k\to\infty} \left\|\nabla_h U(h_k,\cdot)\right\| =0$.  
  %
  %
  

%
\section{Sparse Projections in the Function Space}\label{sec_komp}
As observed before, the update \eqref{eqn_stochastic_update} requires the introduction of a new element $\kappa(s_{T_k},\cdot)$ of the kernel dictionary at each iteration, thus resulting in memory explosion. To overcome this limitation we modify the stochastic gradient ascent by introducing a projection over a RKHS of lower dimension as long as the induced error remains below a given compression budget. This algorithm is known as Orthogonal Match and Pursuit \cite{vincent2002kernel} and we summarize and adapt it to  policy gradient ascent in Algorithm \ref{algorithm_komp}. Starting with the policy $h_0 \equiv 0$, each stochastic gradient ascent iteration  defines a new policy
\begin{equation}\label{eqn_partial_grad_ascent}
\tilde{h}_{k+1} = h_k +\eta\hat{\nabla}_h U(h_k,\cdot),
  \end{equation}
where $\hat{\nabla}_h U(h_k,\cdot)$ is that in \eqref{eqn_stochastic_grad}. 
The difference between the updates \eqref{eqn_partial_grad_ascent} and  \eqref{eqn_gradient_ascent}   is that in \eqref{eqn_partial_grad_ascent} $h_k=\sum_{j=1}^{M_k} w_j^{(k)}\kappa(s_{j}^{(k)},\cdot)$ is represented by a reduced $M_k\leq k$ number of states $s_j^{(k)}$ and weights $w_j^{(k)}$, as it results from the pruning procedure below, (cf., $M_k=k$ for $h_{k+1}$ in \eqref{eqn_gradient_ascent}).

With state $s_{T_k}$ being $s_T$ in step 8 of Algorithm \ref{alg_stochastic_grad}, and $\tilde{w}_{k}:=\eta \frac{\hat{Q}(s_{T_k},a_{T_k};h_k)-\hat{Q}(s_{T_k},\bar{a}_{T_k};h_k)}{2(1-\gamma)}\Sigma^{-1}(a_{T_k}-h_k(s_{T_k}))$, we rewrite   \eqref{eqn_stochastic_grad} as in $\eta\hat{\nabla}_h U(h_k,\cdot) = \tilde{w}_k\kappa(s_{T_k},\cdot)$, and thus
%
%
\begin{equation}
\tilde{h}_{k+1} = \sum_{j=1}^{M_{k}} w_j^{(k)} \kappa(s_{j}^{(k)},\cdot)+\tilde{w}_k\kappa(S_{T_k},\cdot).
\end{equation}
 Hence, $h_k$ is represented by dictionary $D_k=[s_1^{(k)},\ldots, s_{M_k}^{(k)}]$ and associated weights $\bbw_k=\left[\left(w_1^{(k)}\right)^\top,\ldots,\left(w_{M_k}^{(k)}\right)^\top\right]^\top,$ and 
%
%
 $\tilde{h}_{k+1}$ is represented by the  updated  $\tilde{D}_{k+1} = [D_k, s_{T_k}]$ and $\tilde{\bbw}_{k+1} = [\bbw_k^\top, \tilde{w}_{k}^\top]^\top$, which has model order $\tilde{M}_{k+1} = M_k+1$. Then, to avoid memory explosion, we prune the dictionary as long as the induced error stays below a prescribed bound $\epsilon>0$. 
We start by storing copies of the previous dictionary, i.e., define ${D}_{k+1} = \tilde{D}_{k+1}$ and ${\bbw}_{k+1} = \tilde{\bbw}_{k+1}$. Let $\ccalH_{{D}^j_{k+1}}$ be the space spanned by all the elements of ${D}_{k+1}$ except for the $j$-th one. For each $j=1 \ldots {M}_{k+1}$ we identify the less informative dictionary element by solving
%
\begin{align}\label{eqn_pruning_optimization}
    e_j &= \min_{{h}\in\ccalH_{{D}^j_{k+1}}}\left\| {h}-\tilde{h}_{k+1}\right\|_\ccalH^2  =\bbc_j\\
   & +\hspace{-0.1cm} \min_{\bbw \in \mathbb{R}^{p M_{k+1}-1}}\hspace{-0.1cm}
    \bbw^\top \bbK_{{D}^j_{k+1},{D}^j_{k+1}} \bbw -2 \bbw^\top \bbK_{{D}^j_{k+1},\tilde{D}_{k+1}}\tilde{\bbw}_{k+1}, \nonumber 
\end{align}
%
%
%
which results from expanding the square after substituting $h$ and $\tilde{h}_{k+1}$ by their representations as weighted sums of kernel elements, and upon defining the block matrices $\bbK_{{D}^j_{k+1},{D}^j_{k+1}}$ and $\bbK_{{D}^j_{k+1},\tilde{D}_{k+1}}$ whose $(l,m)$-th blocks of size $p\times n$ are $\kappa( s_l^{(k)},s_m^{(k)})$ and $\kappa(s_l^{(k)}, \tilde{s}_m^{(k)})$, respectively,  with $ s_l^{(k)}$ and $s_m^{(k)}$ being the $l$-th and $m$-th elements of $D_{k+1}^{j},$ and with $\tilde{s}_l^{(k)}$ correspondingly  in  $\tilde{D}_{k+1}$.  
%
%
\begin{algorithm}[t]
 \caption{Kernel Orthogonal Matching Pursuit (KOMP)}
\begin{algorithmic}[1]
 \renewcommand{\algorithmicrequire}{\textbf{Input:}}
 \renewcommand{\algorithmicensure}{\textbf{Output:}}
 \Require function ${\tilde{h}_k}$ defined by Dictionary $\tilde{D}_k\in\mathbb{R}^{n\times \tilde{M}_k}$ weights $\tilde{\bbw}_k\in\mathbb{R}^{p \times \tilde{M}_k}$ and compression budget $\epsilon>0$
 \State \textit{Initialize}: ${D}_k = \tilde{D}_k$, ${W}_k=\tilde{W}_k$, $M_k=\tilde{M}_k$, $e^*=0$
 \While{$e^*<\epsilon$ and $M_k>0$}
 \For{$j=1\ldots {M}_k$}
 \State Find minimal error $e_j$ by solving \eqref{eqn_pruning_optimization}
 \EndFor
  \State Less informative element $j^* = \argmin_{j} e_j$ 
  \State Save error $e^* = e_{j^*}$
  \If{Error smaller than compression budget $e^* <\epsilon$}
  \State Prune Dict., ${D}_{k} \leftarrow {D}_{k}^{j^*} $, ${M}_k \leftarrow {M}_k-1$
  \State Update Weights as in \eqref{eqn_least_square_solution}
  $$
{\bbw}_k =  \bbK_{{D}^j_{k},{D}^j_{k}}^\dagger \bbK_{{D}^j_{k},\tilde{D}_{k}}\tilde{\bbw}_{k}
 $$
  \EndIf
  \EndWhile\\
 \Return ${D}_k,{\bbw}_k$ 
 \end{algorithmic}\label{algorithm_komp}
 \end{algorithm}
%
Problem \eqref{eqn_pruning_optimization} is a least-squares problem with the following closed-form solution 
\begin{equation}\label{eqn_least_square_solution}
\bbw^*_j = \bbK_{{D}^j_{k+1},{D}^j_{k+1}}^\dagger \bbK_{{D}^j_{k+1},\tilde{D}_{k+1}}\tilde{\bbw}_{k+1},
  \end{equation}
where, $\left(\cdot\right)^\dagger$ denotes the Moore-Penrose pseudo-inverse. After computing all compression errors $e_j$ we chose the dictionary element that yields the smallest error $j^* = \argmin_{j=1 \ldots M_{k+1}} e_j$, we remove the $j^*$-th column from the dictionary ${D}_{k+1}$, i.e., we redefine ${D}_{k+1} = {D}_{k+1}^{j^*}$ and the model order $M_{k+1} =M_{k+1}-1$ and update the corresponding weights as ${\bbw}_{k+1} = \bbw^*_{j^*}$. We repeat the process as long as the minimum compression error remains below the compression budget, i.e.,  $\min_{j=1 \ldots M_{k+1}} e_j<\epsilon$. The output of the pruning process is a function $h_{k+1}$ that is represent by at most the same number of elements than $\tilde{h}_{k+1}$ and such that the error introduced in this approximation is, by construction, smaller than the compression budget $\epsilon$. This output can be interpreted as a projection over a RKHS of smaller dimension. Let $D_{k+1}$ be the dictionary that Algorithm \ref{algorithm_komp} outputs. Then, the resulting policy can be expressed as 
\begin{equation}\label{eqn_proj_stochastic_ascent}
\hspace{-.1cm}h_{k+1}\hspace{-.1cm} =\hspace{-.1cm} \ccalP_{\ccalH_{D_{k+1}}}\hspace{-.1cm}\left[ \tilde{h}_{k+1}\right]\hspace{-.1cm} =\hspace{-.1cm} \ccalP_{\ccalH_{D_{k+1}}}\hspace{-.1cm}\left[h_{k}+\eta \hat{\nabla}_h U(h_k,\cdot)\right],
\end{equation}
where the operation $\ccalP_{\ccalH_{D_{k+1}}}\left[ \cdot \right]$ refers to the projection onto the RKHS spanned by the dictionary $D_{k+1}$. The algorithm described by \eqref{eqn_partial_grad_ascent} and \eqref{eqn_proj_stochastic_ascent} is summarized in Algorithm \ref{alg_proj_policy_gradient}. By projecting over a smaller subspace we control the model order of the policy $h_k$. However, the induced error translates into an estimation bias on the estimate of $\nabla_h U(h,\cdot)$ as we detail in the next proposition 

  \begin{algorithm}[t]
 \caption{Projected Stochastic Policy Gradient Ascent}
\begin{algorithmic}[1]
 \renewcommand{\algorithmicrequire}{\textbf{Input:}}
 \renewcommand{\algorithmicensure}{\textbf{Output:}}
 \Require step size $\eta$, compression budget $\epsilon$
 \State \textit{Initialize}: $h_0=0$
 \For{$k=0 \ldots$}
 \State Compute $\hat{\nabla}_h U(h_k,\cdot) =$ StochasticGradient($h_k$)
 \State Update policy via stochastic gradient ascent
 $$
 \tilde{h}_{k+1} = h_k +\eta \hat{\nabla}_h U(h_k,\cdot)
 $$
 \State Reduce model order $h_{k+1} =$ KOMP($\tilde{h}_{k+1},\epsilon$) 
 \EndFor
 \end{algorithmic}\label{alg_proj_policy_gradient}
 \end{algorithm}

\begin{proposition}\label{prop_error_bound}
  The update of Algorithm \ref{alg_proj_policy_gradient} is equivalent to running biased stochastic gradient ascent, with bias
\begin{equation}\label{eqn_bias}
b_k = \ccalP_{\ccalH_{D_{k+1}}}\left[h_{k}+\eta \hat{\nabla}_h U(h_k,\cdot)\right] -\left(h_{k}+\eta \hat{\nabla}_h U(h_k,\cdot)\right). 
  \end{equation}
  bounded by the compression budget $\varepsilon$ for all $k$.    
	%
\end{proposition}
\begin{proof}
From \eqref{eqn_proj_stochastic_ascent} and adding and subtracting $\eta \hat{\nabla}_h U(h_k,\cdot)$, it is possible to write the difference $h_{k+1}-h_k$ as
\begin{equation}
  \begin{split}
    h_{k+1} - h_k  &=\ccalP_{\ccalH_{D_{k+1}}}\left[h_{k}+\eta \hat{\nabla}_h U(h_k,\cdot)\right] \\
    &-\left(h_{k}+ \eta \hat{\nabla}_h U(h_k,\cdot)\right) +\eta \hat{\nabla}_h U(h_k,\cdot).
  \end{split}
\end{equation}
Using the definition of the bias \eqref{eqn_bias} the previous expression can be written as
\begin{equation}\label{eqn_bk}
h_{k+1} = h_k +\eta \hat{\nabla}_h U(h_k,\cdot)+b_k.
\end{equation}
To complete the proof, notice that by definition $b_k$ is the error of the compression and thus its norm is bounded by the compression budget $\varepsilon$. 
  \end{proof}
As stated by the previous proposition the effect of introducing the KOMP algorithm is that of updating the policy by running gradient ascent, where now the estimate  is biased. Hence, we claim in the following result that Stochastic Policy Gradient Ascent (Algorithm \ref{alg_proj_policy_gradient}) converges to a neighborhood of a critical point of the expected discounted reward, whose size depends on the step-size of the algorithm as well as on compression error allowed. However, whereas the model order of the function obtained via stochastic gradient ascent without projection (Algorithm \ref{alg_policy_gradient}) could grow without bound, for the projected version we can ensure that the model order obtained is always bounded. We formalize these results next. 
\begin{theorem}\label{theo_parsimonious}
  Let $\eta>0$ and $\epsilon>0$ for all $k\geq 0$. Then there exists a constant $C:=C(\gamma,\eta,\epsilon,\Sigma,B_r,)$ such that 
  \begin{equation}\label{eqn_liminf}
    \liminf_{k\to\infty} \left\|\nabla_h U(h_k,\cdot)\right\|_{\ccalH} \leq \frac{\epsilon}{2\eta}+\frac{\sqrt{\epsilon^2+4\eta^3C}}{2\eta},
    \end{equation}
with probability one. Moreover, there exists a constant $M:= M(\epsilon)>0$ such that for every $k\geq 0$ the model order $M_k$ needed to represent the function $h_k$ is such that $M_k \leq M$.  
\end{theorem}
\begin{proof}
The proof of this result is the matter of Section \ref{sec_convergence_pars}.
  \end{proof}
%
%
Observe that the optimal selection is $\epsilon=O(\eta^{3/2})$ in the sense that selecting
a smaller compression factor, the total error bound is of $O(\eta^{3/2})$. In that sense, such selection is not optimal, because we force a small compression error -- which entails larger model order -- and there is no benefit in terms of the convergence error. Then the parameter $\eta$ is to be chosen trading-off accuracy for speed of convergence.


\section{Convergence Analysis of Sparse Policy Gradient}\label{sec_convergence_pars}

This section contains the proof of Theorem \ref{theo_parsimonious}. It starts by providing a lower bound on the expectation of random variables $U(h_{k+1})$ conditioned to the sigma field $\ccalF_k$ 
\begin{lemma}\label{lemma_supermartingale_pars}
   The sequence of random variables $U(h_k)$ satisfies the following inequality 
\begin{equation}\label{eqn_supermartingale_equation_pars}
    \begin{split}
      \mathbb{E}\left[{U(h_{k+1})|\ccalF_k}\right] &\geq  U(h_k)- \eta^2 C\\
      &\hspace{-1cm}+\eta\left\|\nabla_h U(h_k)\right\|_\ccalH\left(\left\|\nabla_h U(h_k)\right\|_\ccalH -\frac{\epsilon}{\eta}\right), 
\end{split}
\end{equation}
where $C$ is the following positive constant
\begin{equation}\label{eqn_bound_for_radius_convergence}
C = L_1\left(\sigma^2+2\frac{\epsilon}{\eta}\sigma +\frac{\epsilon^2}{\eta^2}\right)+\eta L_2\left(\sigma^2+2\frac{\epsilon}{\eta}\sigma +\frac{\epsilon^2}{\eta^2}\right)^{3/2},
\end{equation}
where  $L_1$ and $L_2$ are the constants defined in Lemma \ref{lemma_lipchitz} and $\sigma$ is the constant defined in 
Lemma \ref{lemma_stochastic_gradient}.
  \end{lemma}
\begin{proof}
  As in the proof of Lemma \ref{lemma_supermartingale} we can lower bound $\mathbb{E}[U(h_{h+1})|\ccalF_k]$ by (cf., \eqref{eqn_first_supermartingale_bound})
  \begin{equation}\label{eqn_first_supermartingale_bound_pars}
    \begin{split}
      \mathbb{E}\left[{U(h_{k+1})|\ccalF_k}\right] \geq U(h_k)+\left<\nabla_h U(h_k),\mathbb{E}\left[h_{k+1}-h_k|\ccalF_k\right]\right>_{\ccalH} \\
      -L_1\mathbb{E}\left[\left\|h_{k+1}-h_k\right\|_{\ccalH}^2|\ccalF_k\right]
      -L_2\mathbb{E}\left[\left\|h_{k+1}-h_k\right\|_{\ccalH}^3|\ccalF_k\right].
      \end{split}
    \end{equation}

  Substitute \eqref{eqn_bk} for $h_{k+1}$  to write the expectation of the quadratic term in the right hand side of \eqref{eqn_first_supermartingale_bound_pars} as 
  \begin{align*}
         L_1\mathbb{E}\left[\left\|h_{k+1}-h_k\right\|_{\ccalH}^2|\ccalF_k\right] 
     \leq  L_1\left(\eta^2 \mathbb{E}\left[\left\|\hat{\nabla}_h U(h_k,\cdot)\right\|_{\ccalH}^2|\ccalF_k\right]\right.\\
		\left.+2\eta\epsilon\mathbb{E}\left[\left\|\hat{\nabla}_h U(h_k,\cdot)\right\|_{\ccalH}|\ccalF_k\right]+\epsilon^2 \right),
 \end{align*}
  where we have used that $\|b_k\|\leq \epsilon$ as stated in Proposition \ref{prop_error_bound}. Using the bounds provided in  Lemma \ref{lemma_stochastic_gradient}, the previous expression can be upper bounded by
  \begin{equation}\label{eleuno}
    L_1\mathbb{E}\left[\left\|h_{k+1}-h_k\right\|_{\ccalH}^2|\ccalF_k\right] \leq \eta^2L_1\left(\sigma^2+2\frac{\epsilon}{\eta}\sigma +\frac{\epsilon^2}{\eta^2}\right).
    \end{equation}
  With a similar procedure we obtain
  \begin{equation}\label{eledos}
    \begin{split}
      L_2\mathbb{E}\left[\left\|h_{k+1}-h_k\right\|_{\ccalH}^3|\ccalF_k\right] \leq \eta^2 \eta_0L_2\left(\sigma^2+2\frac{\epsilon}{\eta}\sigma +\frac{\epsilon^2}{\eta^2}\right)^{3/2}.
      \end{split}
    \end{equation}
  Observe that the sum of \eqref{eleuno} and \eqref{eledos} is equal to $\eta^2C$ in \eqref{eqn_bound_for_radius_convergence}. Then, substitute \eqref{eleuno} and \eqref{eledos} in  \eqref{eqn_first_supermartingale_bound_pars} to obtain
  \begin{equation}
    \begin{split}\label{eqn_martigale_innerproduct}
      \mathbb{E}\left[{U(h_{k+1})|\ccalF_k}\right] &\geq U(h_k)-C\eta^2\\
&      +\left<\nabla_h U(h_k),\mathbb{E}\left[h_{k+1}-h_k|\ccalF_k\right]\right>_{\ccalH}. 
\end{split}
    \end{equation}
  Finally, \eqref{eqn_supermartingale_equation_pars} results from applying the Cauchy-Schwartz inequality   to the inner product in \eqref{eqn_martigale_innerproduct} and then substituting \eqref{eqn_bk} for $h_{k+1}$, with $\|b_k\|\leq \epsilon$. 

  %
  %
  %
  \end{proof}
The previous Lemma establishes a lower bound on the expectation of $U(h_{k+1})$ conditioned to the sigma algebra $\ccalF_k$. This lower bound however, is not enough for $U(h_k)$ to be a submartingale, since the sign of the term added to $U(h_k)$ in the right hand side of \eqref{eqn_supermartingale_equation_pars} depends on the norm of $\nabla_h U(h_k)$. This is in contrast with the situation in Lemma \ref{lemma_supermartingale}, where the term was always positive. The origin of this issue lies on the bias introduced by the sparsification. However, when the norm of the gradient is large the term is negative and we have a submartingale outside of a neighborhood of the critical points. To formalize this idea let us define the neighborhood as
 \begin{equation}\label{neigborhood_inequality}
\left\|\nabla_h U(h_k,\cdot)\right\|_{\ccalH} \leq \frac{\epsilon}{2\eta}+\frac{\sqrt{\epsilon^2+4\eta^3C}}{2\eta},
\end{equation}
 and the corresponding stopping time
%
  %
  %
  %
  %
  \begin{equation}
N = \min_{k\geq 0} \left\{\left\|\nabla_h U(h_{k},\cdot)\right\|_{\ccalH} \leq \frac{\epsilon}{2\eta}+\frac{\sqrt{\epsilon^2+4\eta^3C}}{2\eta}  \right\}.
    \end{equation}
  
	In order to prove  \eqref{eqn_liminf} we will argue that either the limit exists and satisfies the bound in \eqref{eqn_liminf}, or $P(N<\infty)=1$, in which case \eqref{neigborhood_inequality} must be recursively satisfied after a finite number of iterations so that  \eqref{eqn_liminf} holds.  
  In this direction we define $V_k=\left(U(h^*)-U(h_{k})\right)\mathbbm{1}(k\leq N)$, with $\mathbbm{1}(\cdot)$ being the indicator function, and prove that $V_k$ is a non-negative submartingale. Indeed, since $U(h^*)$ maximizes $U(h)$, $V_k$ is always non-negative. In addition $V_k\in \ccalF_k$ since $U(h_k)\in\ccalF_k$ and $\mathbbm{1}(k-1\leq N)\in \ccalF_{k}$. To show that $\mathbb{E}[V_{k+1} | \ccalF_k] \leq V_k$ start by using that $\mathbbm{1}(k\leq N)\in\ccalF_{k}$ and write
  \begin{equation}
\mathbb{E}\left[V_{k+1} | \ccalF_k\right]  =  \mathbbm{1}(k+1\leq N) \mathbb{E}\left[U(h^*)-U(h_{k+1}) | \ccalF_k\right]. 
    \end{equation}
  Using \eqref{eqn_supermartingale_equation_pars} and defining the following variable
  \begin{align}
W_k &:= \eta \left\|\nabla_hU(h_k,\cdot) \right\|_{\ccalH}^2 - \epsilon \left\|\nabla_hU(h_k,\cdot) \right\|_{\ccalH} - \eta^2 C   \label{eqn_def_wk},
  \end{align}
  we can upper bound $\mathbb{E}\left[V_{k+1} | \ccalF_k\right]$ as
  \begin{align}\label{eqn_supermartingale_equation_pars_stop}
  \mathbb{E}\left[V_{k+1} | \ccalF_k\right]  \leq  \mathbbm{1}(k+1\leq N)\left(  \left(U(h^*)-\hspace{-.1cm}U(h_{k})\right)- W_k\right).
\end{align}
Notice that the bound in \eqref{neigborhood_inequality} is root of \eqref{eqn_def_wk} as a polynomial in the variable $\left\|\nabla_hU(h_k,\cdot) \right\|$. It follows that   $W_k>0$ as long as $k< N$, so that $\mathbbm{1}(k+1 \leq N)W_k \geq 0$ for all $k$. Also notice that the indicator function $\mathbbm{1}(k\leq N)$ is non-increasing with $k$, so that $\mathbbm{1}(k+1 \leq N)\leq  \mathbbm{1}(k+1 \leq N)$. Using these two facts, it follows from \eqref{eqn_supermartingale_equation_pars_stop} that $\mathbb{E}[V_{k+1} | \ccalF_k] \leq V_k$. Thus, $V_k$ is a nonnegative submartingale and therefore it converges to random variable $V$ such that $\mathbb{E}[V]\leq \mathbb{E}[V_0]$ (see e.g., \cite[Theorem 5.29]{durrett2010probability}).
Rearranging the terms in \eqref{eqn_supermartingale_equation_pars_stop} and considering the total expectation we have that 
\begin{equation}\label{eqn_supermartingale_equation_pars_before_limit}
  \begin{split}
   \mathbb{E}\left[\sum_{j = 0}^{k} \mathbbm{1}(j < N) W_k  \right]    \leq \mathbb{E}[V_0] -\mathbb{E}\left[ V_{k+1} \right] .
   \end{split}
    \end{equation}
 Again, by definition of the stopping time $N$, $\mathbbm{1}(k<N) W_k$ is nonnegative, and thus the sequence of random variables
 %
 $\sum_{j = 0}^{k}\mathbbm{1}(j < N)W_j$,
 %
 is monotonically increasing. Hence, the Monotone Convergence Theorem (see e.g.,\cite[Theorem 1.6.6]{durrett2010probability}) allows us to write 
 \begin{equation}\label{eqn_supermartingale_equation_pars_limit1}
   \begin{split}
     \lim_{k\to\infty} \mathbb{E}\left[\sum_{j = 0}^{k}\mathbbm{1}(j < N)W_j \right] =\mathbb{E}\left[\sum_{j = 0}^{\infty}\mathbbm{1}(j < N)W_j  \right].
\end{split}
   \end{equation}
On the other hand, $U(h_k)$ is bounded according to Lemma \ref{lemma_bounded}, thus  $V_k$ is a bounded sequence and then we use the Dominated Convergence Theorem (see e.g. \cite[Theorem 1.6.7]{durrett2010probability}) to obtain 
  \begin{equation}\label{eqn_supermartingale_equation_pars_limit2}
\mathbb{E}[V] = \mathbb{E}[\lim_{k\to\infty} V_k] = \lim_{k\to \infty} \mathbb{E}[V_k].
    \end{equation}
  Taking the limit of $k$ going to infinity in both sides of \eqref{eqn_supermartingale_equation_pars_before_limit} and using \eqref{eqn_supermartingale_equation_pars_limit1} and \eqref{eqn_supermartingale_equation_pars_limit2} we have that 
  \begin{equation}\begin{split}
    \mathbb{E}\left[\sum_{j = 0}^{\infty}\mathbbm{1}(j < N)W_j  \right] \leq \mathbb{E}[V_0]-\mathbb{E}[V] < \infty.
    \end{split}
    \end{equation}
  Observe that the expectation on the left hand side of the previous expression can be computed as
  \begin{equation}
    \begin{split}
      P(N\hspace{-.1cm}<\infty)      \mathbb{E}\hspace{-.1cm}\left[\sum_{j = 0}^{N-1}W_j  \Big|_{ N<\infty} \right] 
      \hspace{-.1cm}+\hspace{-.1cm} P(N=\infty)\mathbb{E}\hspace{-.1cm}\left[\sum_{j = 0}^{\infty}W_j \Big|_{N=\infty} \right]. 
\end{split}
    \end{equation}
 By virtue of Lemma \ref{lemma_bounded_gradient}, $\left\| \nabla_h U(h,\cdot) \right\|$ is uniformly bounded for all $h\in\ccalH$. Thus, the first sum  in the previous expression is finite. Hence, 
  \begin{equation} 
P(N=\infty)      \mathbb{E}\left[\sum_{j = 0}^{\infty}W_j\Big| N=\infty  \right]  <\infty.
  \end{equation}
  The latter can only hold if $P(N=\infty) = 0$ or if the expectation of the sum is bounded. If the former happens it means that infinitely often $\left\| \nabla_h U(h,\cdot) \right\|$  visits the neighborhood  \eqref{neigborhood_inequality},   and thus  \eqref{eqn_liminf} holds. It remains to analyze the case where the expectation of the sum is finite. Using the Monotone Convergence Theorem one can exchange the expectation with the sum and therefore we have that
  %
$\sum_{j = 0}^{\infty}    \mathbb{E}\left[W_j\Big| N=\infty  \right]  <\infty,$
  %
  which implies that $\lim_{k\to\infty} \mathbb{E}[W_k|N=\infty] =0$. Thus 
  \begin{equation}
    \lim_{k\to\infty}\mathbb{E}\left[\left(\eta\left\|\nabla_h U(h_{k},\cdot)\right\|^2-\left\| \nabla_h U(h_{k},\cdot)\right\|\epsilon-\eta^2 C\right)\right] = 0.
    \end{equation}
  Moreover, because the norm of the gradient is bounded, the Dominated Convergence Theorem allows us to write
    \begin{equation}
    \mathbb{E}\left[\lim_{k\to\infty}\left(\eta\left\|\nabla_h U(h_{k},\cdot)\right\|^2-\left\| \nabla_h U(h_{k},\cdot)\right\|\epsilon-\eta^2 C\right)\right] = 0.
    \end{equation}
    Because the random variable is nonnegative it must hold that 
  \begin{equation}
\lim_{k\to\infty} \left\|\nabla_h U(h_k,\cdot)\right\|_{\ccalH} = \frac{\epsilon}{2\eta}+\frac{\sqrt{\epsilon^2+4\eta^3C}}{2\eta}.
    \end{equation}
  Thus, \eqref{eqn_liminf} holds as well if $P(N=\infty)>0$. The proof that the model order of the representation is bounded for all $k$ is identical to that in \cite[Theorem 3]{koppel2016parsimonious}.
%


\section{Numerical Experiments}\label{sec_numerical}
In sections \ref{sec_mountain_car} and \ref{sec_cartpole} we test 
Algorithm \ref{alg_proj_policy_gradient} in the problems of the mountain car and the cartpole. 
\subsection{Mountain Cart}\label{sec_mountain_car}
We benchmarked Stochastic Projected Policy Gradient Ascent on a classic control problem, the Continuous Mountain Car \cite{argyriou2009there}, which is featured in OpenAI Gym \cite{mountaincar}. In this problem, the $n=2$ dimensional state space consists of position and velocity, bounded within $[-1.2, 0.6]$ and $[-0.07, 0.07]$, respectively. The action space is a scalar representing the real valued force on the car. The reward function is 100 when the car reaches the goal at position $0.6$, and in every episode it subtracts $0.1\sum_{t=t_0}^{t_f} a_t^2$, where $a_t$ are the actions selected. Because of the penalization of the actions, in the space of policies there are local maxima around policies that keep the car stationary in order to realize roughly zero reward. In order to avoid converging to such policy, we set $h_0$ to have kernels at $(0.65,-0.02)$ and $(-0.35,0.02)$ with respective weights $0.5$ and $-0.5$. In particular, we work with Gaussian kernels, that are nonsymmetric due to the difference in the scales of position and velocities attained by the mountain cart. Their covariance matrix is given by $\diag( [0.15, 0.015])$. 
\begin{figure}
  \centering
    \includegraphics[width=1\linewidth]{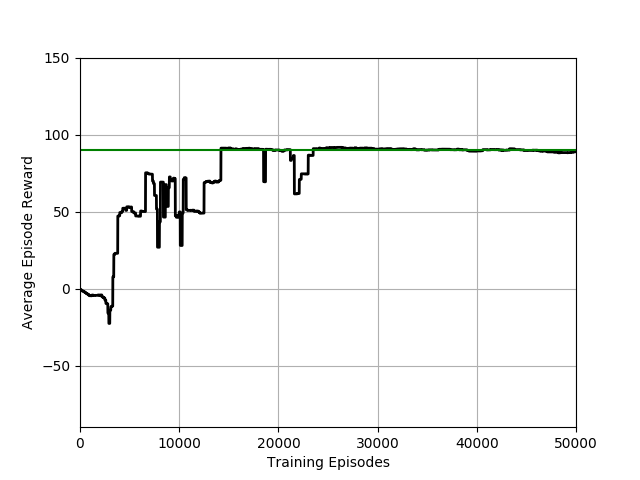}\\
    \includegraphics[width=1\linewidth]{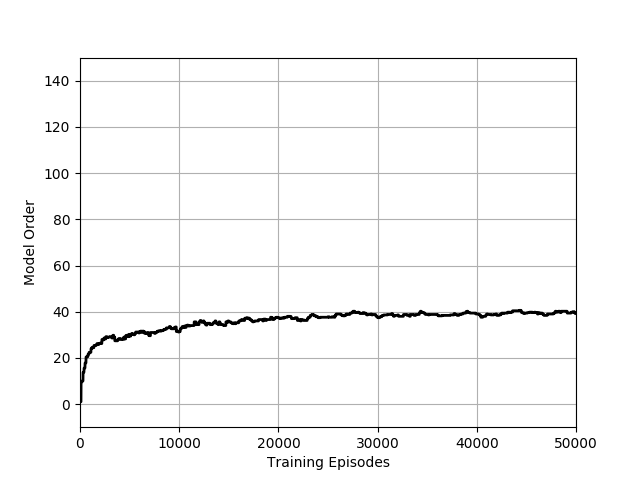}
    \caption{Result of representative run of Algorithm \ref{alg_proj_policy_gradient} over 50000 Continuous Mountain Car episodes.  The top figure shows the average reward obtained by the policy --showed in Figure \ref{fig_policy}-- after each training step (episode).  An average reward over 90 (green) indicates that we have solved the problem, reaching the goal location. The bottom figure shows the model complexity (number of Dictionary elements) during training remains bounded.}\label{fig_mountain_car}
\end{figure}
\begin{figure}
	\centering
		\includegraphics[width=1.0 \linewidth]{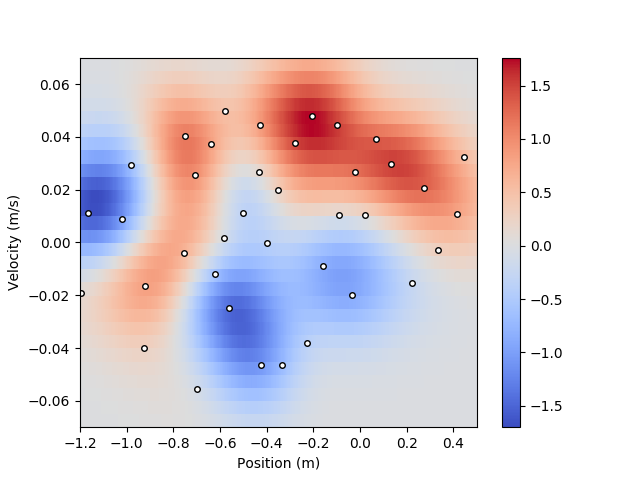}
	\caption{Learned policy for Continuous Mountain Car after 50000 episodes.}\label{fig_policy}
\end{figure}
%
The results obtained with Algorithm \ref{alg_proj_policy_gradient} for the following parameters: $\gamma =0.999$, $\Sigma = 1.3$, $\eta = 0.0005$ and  $\epsilon=0.000335$ are given in Figs. \ref{fig_mountain_car} and \ref{fig_policy}. Fig. \ref{fig_mountain_car} shows the average reward during training (top), and the model order (bottom). The policy learned  after 50000 training samples is given in Fig. \ref{fig_policy}. From Fig. \ref{fig_mountain_car} we can observe that the policy converges to a solution that solves the problem in about 25000 samples.  
The challenge in the mountain car is that it is not possible to escape the valley by just accelerating to the right. Hence, the optimal policy needs to be such that the cart oscillates to increase its velocity. In particular, in Fig. \ref{fig_policy} we observe that for positive speeds the acceleration is mostly positive, while when the speed is negative, so is the force.

In contrast to other Kernel based RL algorithms, such as \cite{tolstayanonparametric}, ours manages to significantly reduce the  computational complexity by only updating the dictionary after a sequence of actions. In practice, our algorithm performs cheap actions (as measured by time and computational complexity) in order to perform relatively few computationally intensive learning steps. In particular, the most costly subroutine is KOMP (Algorithm \ref{algorithm_komp}) and we resource to it only once per episode.

\subsection{Cartpole Problem}\label{sec_cartpole}

We also tested the algorithm for the Cartpole problem, featured in OpenAI Gym \cite{cartpole}. In this case the state space is $n=4$ dimensional, consisting of position and velocity of the cart, and angle and angular speed of the pole. The position and the angle are bounded respectively within $[-2.4, 2.4]$ and $[-41.8^\circ, 41.8^\circ]$ while the velocity and the angular velocity are unbounded. The action space is either to apply a fixed force either to the left or to the right of the cart.  A successful trial is one in which the pole is kept vertical within $\pm 0.5^\circ$ for more than 195 steps. As in the previous section, we work with Gaussian kernels, that are nonsymmetric due to the difference in the scales of position and velocities attained by the mountain cart. Their covariance matrix is given by $\diag( [0.3, 0.1, 0.1, 0.1])$. 
\begin{figure}
  \centering
    \includegraphics[width=1\linewidth]{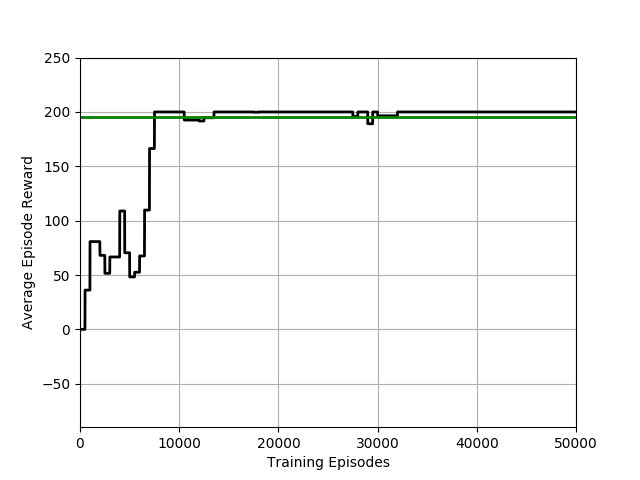}\\
    \includegraphics[width=1\linewidth]{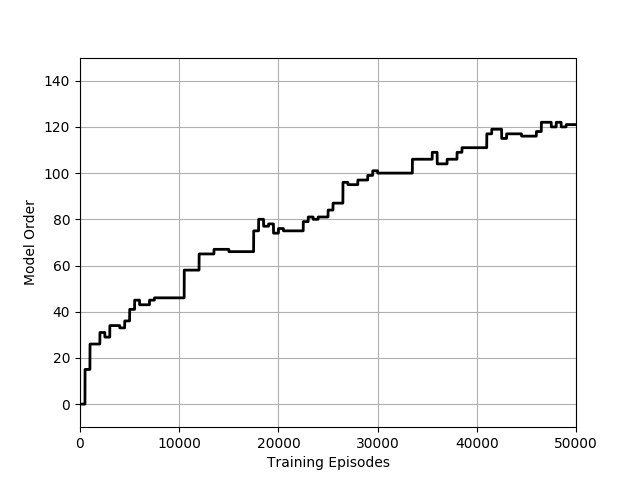}
    \caption{Result of representative run of Algorithm \ref{alg_proj_policy_gradient} over 50000 Cartpole episodes.  The top figure shows the average reward obtained by the policy --showed in Figure \ref{fig_policy_cartpole}-- after each training step (episode).  An average reward over 195 (green) indicates that we have solved the problem, holding the pole vertical for long enough. The bottom figure shows the model complexity (number of Dictionary elements) during training remains bounded.}\label{fig_cartpole}
\end{figure}
\begin{figure}
  \centering
    \includegraphics[width=1\linewidth]{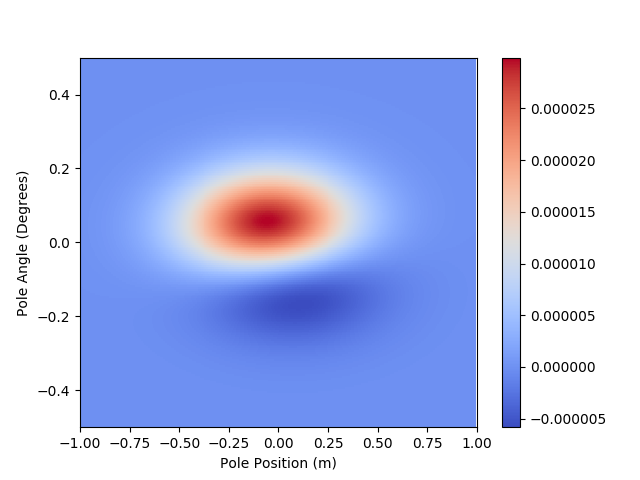}\\
    \includegraphics[width=1\linewidth]{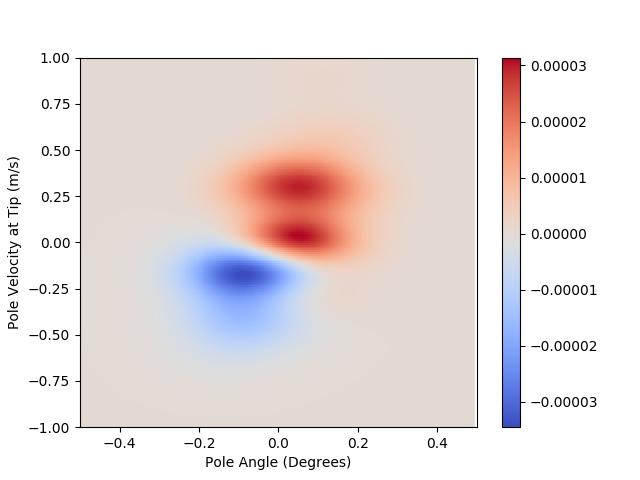}
    \caption{Learned policy for the cartpole problem after 50000 training episodes. In the above figure we observe the policy for the position and angle. The latter states, that if the pole tilts right --positive angles-- the cart should accelerate to the right. Likewise, if the pole tilts left, the cart needs to accelerate left. In the figure below We observe that the policy is such that for positive angles and angular speeds, the cart will accelerate to the right. These policies are intuitive from a physics standpoint. }\label{fig_policy_cartpole}
\end{figure}
%
%
The results obtained with Algorithm \ref{alg_proj_policy_gradient} for the following paramters: $\gamma =0.95$, $\Sigma = 0.01$, $\eta = 0.00005$ and  $\epsilon=8.839 10^{-7}$ are given in figures \ref{fig_cartpole} and \ref{fig_policy_cartpole}. In the former, we plot the average reward during training (top figure), and the model order (bottom figure). The policy learned from this experiment is given in Figure \ref{fig_policy_cartpole}, where we plot the policy learned after after 50000 iterations. From \ref{fig_cartpole} we can observe that the policy converges to a solution that allows to solve the problem in about 8000 training examples. According to the learned policy, the cart accelerates to the right when the pole is tilting right and to the left when tilting left, corroborating the capability of Algorithm \ref{alg_proj_policy_gradient} to unfold an intuitive behavior. 


\section{Conclusion}
We have considered the problem of learning a policy on a RKHS that maximizes the expected cumulative reward of an agent. In particular, we presented an algorithm to obtain an unbiased estimate of the gradient of the reward with respect to the policy. By running stochastic gradient ascent in the RKHS we were able to show convergence to a critical point of the reward. This algorithm, is not practical since the number of kernel elements that requires grows unbounded. To overcome this limitation, we combined the previous algorithm with destructive KOMP to ensure that the model order remains bounded. Reducing numerical complexity is traded-off for accuracy, with  a theoretical result that guarantees convergence to a neighborhood of the critical points.

\appendix

In this appendix we present some properties of the expected discounted reward and its gradient which are needed in the convergence analysis of functional stochastic gradient ascent. 
%
\begin{lemma}\label{lemma_bounded}
  Under Assumption \ref{assumption_reward_function} the expected discounted reward defined in \eqref{eqn_problem_statement} and the Q-function defined in \eqref{eqn_q_function} satisfy
  \begin{equation}
    \left| U(h) \right| < \frac{B_r}{1-\gamma} \quad \mbox{and} \quad  
    \left| Q(s,a;h) \right| < \frac{B_r}{1-\gamma} \quad \forall \quad h\in\ccalH.
  \end{equation}
\end{lemma}
\begin{proof}
  The triangle inequality applied to $|U(h)|$ yields  
  \begin{equation}
\left|U(h)\right| \leq \mathbb{E}\left[\left|\sum_{t=0}^\infty \gamma^tr(s_t,a_t)\right|\Big|h\right] \leq  \mathbb{E}\left[\sum_{t=0}^\infty \gamma^t|r(s_t,a_t)|\Big|h \right],
  \end{equation}
   By virtue of Assumption \ref{assumption_reward_function}, $|r(s,a)|<B_r$ for all $(s,a) \in \ccalS\times\ccalA$, hence it follows that
  %
$\left|U(h)\right| \leq B_r \sum_{t=0}^\infty \gamma^t = {B_r}/({1-\gamma}).$
  The proof of the result for $Q(s,a;h)$ is analogous. 
  \end{proof}
%

\begin{lemma}\label{lemma_bounded_gradient}
Let Assumption \ref{assumption_reward_function} hold, then $\nabla_h U(h,\cdot)$ defined as in \eqref{eqn_nabla_U} is bounded for all $h\in\ccalH$.
\end{lemma}
\begin{proof}
 Starting from  \eqref{eqn_nabla_U} and considering $\left\|k(s,\cdot)\right\| =1$ (cf., Definition \ref{def_rkhs}), one can write
\begin{equation}\begin{split}
    \left\|\nabla_h U(h,\cdot) \right\|\leq \frac{1}{1-\gamma}\mathbb{E}_{\rho}\left[|Q(s,a;h)|\left\|\Sigma^{-1}\left(a-h(s)\right)\right\|\right],
    \end{split}
\end{equation}
which can be further upper bounded by virtue of Lemma \ref{lemma_bounded} as
\begin{equation}\label{eqn_norm_gradient_partial}
  \left\|\nabla_h U(h,\cdot) \right\|\leq \frac{B_r}{(1-\gamma)^2}\mathbb{E}_{(s,a)\sim \rho(s,a)}\left[\left\|\Sigma^{-1}\left(a-h(s)\right)\right\|\right].
\end{equation}
By construction $\Sigma^{-1/2}(a-h(s))$ is a multivariate normal distribution, hence the expectation of its norm is bounded. 
\end{proof}

\begin{lemma}\label{lemma_lipchitz}
  Let Assumption \ref{assumption_reward_function} hold, with constant $B_r$. Then the gradient of the expected discounted reward satisfies 
  \begin{equation}
    \left\|\nabla_h U(h_1,\cdot) - \nabla_h U(h_2,\cdot)\right\|_\ccalH \leq L_1\|h_1-h_2\|_\ccalH+L_2\|h_1-h_2\|_{\ccalH}^2,\label{eq_bound_lemma5}
  \end{equation}
  for all $h_1,h_2 \in\ccalH$ with $L_1$ and $L_2$ given by 
  \begin{equation*}
    L_1 = B_r \frac{(1-\gamma+p (1+\gamma))}{\lambda_{\min}(\Sigma)(1-\gamma)^3},\     L_2 = B_r\frac{(1+\gamma) \sqrt{p}}{\left(\lambda_{\min}(\Sigma)\right)^{3/2}(1-\gamma)^3}.
  \end{equation*}
%
  %

  \end{lemma}
\begin{proof}
  %
%
%
%
  Consider the following bound to be used later 
	\begin{equation}
         \left\|h(s)\right\| = \left|\left\langle  h,\kappa(s,\cdot)\right \rangle_\ccalH\right| \leq \left\|h\right\|.\label{eqn_bound_hs}
     \end{equation}
		 due to the Cauchy-Scwartz inequality and with $\|\kappa(s,.)\|=1$ (cf., Definition \ref{def_rkhs}). Substituting \eqref{eqn_q_function} for $Q(s,a;h)$  in \eqref{eqn_after_swaping_sum_expectation} yields
	  \begin{equation}
      \hspace{-.1cm}\nabla_h U(h,\cdot) \hspace{-.1cm}=\hspace{-.1cm}\sum_{t=0}^\infty \sum_{u=0}^\infty \hspace{-.1cm}\gamma^{t+u}\mathbb{E}_{p_h}\left[r(s_{t+u},a_{t+u})\kappa(s_t,\cdot)\zeta_t^h\right] \label{eqn_nablau_sumr}
\end{equation}
where we have defined $\zeta^h_t:=\Sigma^{-1}\left(a_t-h(s_t)\right)$ for notational brevity. The expectation in \eqref{eqn_nablau_sumr} is integrated with 
\begin{align}\label{eqn_phsatu}
p_h(\bbs,\bba):= p_{t+u}(\bbs,\bba) \prod_{r=0}^{t+u} \pi_{h_1}(a_r|s_r)
\end{align}
 with $\bbs$ and $\bba$ collecting states and actions up to time $t+u$, and with $p_{t+u}(\bbs,\bba) :=p(s_0)\prod_{r=0}^{t+u-1} p(s_{r+1}|s_r,a_r)$.	Expanding the expectation as an integral and 
%
                 adding and subtracting $\mathbb{E}_{p_{h_2}}\left[r(s_{t+u},a_{t+u})\kappa(s_{t},\cdot) \zeta_t^{h_1} \right]$,  yields   %
\begin{align*}
 &\nabla_hU(h_1,\cdot) - \nabla_h U(h_2,\cdot)  \\
   &=\sum_{t=0}^\infty \sum_{u=0}^\infty \gamma^{t+u}\int r(s_{t+u},a_{t+u})\zeta_t^{h_1} \kappa(s_t,\cdot)p_{t+u}(\bbs,\bba) \\
 &\hspace{2.5cm}\times  \left(\prod_{r=0}^{t+u} \pi_{h_1}(a_r|s_r)  -\prod_{r=0}^{t+u} \pi_{h_2}(a_r|s_r)\right) d\bbs d\bba\\
 &+\sum_{t=0}^\infty \sum_{u=0}^\infty \gamma^{t+u} \int r(s_{t+u},a_{t+u})\Sigma^{-1}\left( h_2(s_t)-h_1(s_t) \right)\\
&\hspace{2.5cm}\times \kappa(s_t,\cdot)p_{h_2}(\bbs,\bba) d\bbs d\bba.
\end{align*}

                  Using that $|r(s_{t+u},a_{t+u})|\leq B_r$ and $\|\kappa(s_t,\cdot)\|=1$ (cf., Assumption \ref{assumption_reward_function} and Definition \ref{def_rkhs}, respectively) we can bound
                  \begin{align}\label{eqn_split_delta}
                        &\left\|\nabla_hU(h_1,\cdot) - \nabla_h U(h_2,\cdot) \right\|\leq \sum_{t=0}^\infty \sum_{u=0}^\infty \gamma^{t+u}B_r (I_1+I_2)
												\end{align}									%
with
\begin{align*}												
           \nonumber         &I_1:=\int \left\|\zeta_t^{h_1}\right\| \left|\Delta_{\pi}(h_1,h_2,s,a)\right| p_{t+u}(\bbs,\bba) d\bbs d\bba\\
                      &I_2:=\int \left\|\Sigma^{-1}\left( h_2(s_t)-h_1(s_t)\right) \right\|                     p_{h_2}(\bbs,\bba) d\bbs d\bba,\\
&\Delta_{\pi}(h_1,h_2,s,a)  := \prod_{r=0}^{t+u}\pi_{h_2}(a_r|s_r)-\prod_{r=0}^{t+u}\pi_{h_1}(a_r|s_r).
                    \end{align*}

 To obtain a bound for $I_1$ in \eqref{eqn_split_delta} define  
 %
%
%
 %
  $h_{\lambda} = \lambda h_1+(1-\lambda)h_2$ with $\lambda\in [0,1]$. Next, consider the Taylor expansion of 	$\prod_{r=0}^{t+u}\pi_{h}(a_r|s_r)$ as a function of $h$, which yields 
       \begin{equation} \begin{split}
\Delta_{\pi}(h_1,h_2,\bbs,\bba) =\sum_{r=0}^{t+u} \left\langle \zeta_r^{h_\lambda}\prod_{r=0}^{t+u}\pi_{h_\lambda}(a_r|s_r) \kappa(s_r,\cdot),h_1-h_2\right\rangle
         \end{split} \end{equation}
     %
     Thus, the Cauchy-Schwartz inequality allows us to write
     \begin{equation}\label{eqn_bound_delta_pi}
\left|\Delta_{\pi}(h_1,h_2,s,a) \right|        \leq \left\|h_1-h_2\right\|\sum_{r=0}^{t+u}\left\|\zeta_r^{h_\lambda}\right\|\prod_{r=0}^{t+u} \pi_{h_\lambda}(a_r|s_r).
       \end{equation}
     With this in mind we bound $I_1$ in \eqref{eqn_split_delta}.  Substituting \eqref{eqn_bound_delta_pi} and using the definition of $p_{h_\lambda}$ in \eqref{eqn_phsatu} yields
      \begin{align}
      \nonumber          &I_1=\int  p_{t+u}(\bbs,\bba)\left\|\zeta_t^{{h_1}}\right\|\left|\Delta_{\pi}(h_1,h_2,s,a)\right|\,d\bbs d\bba\\
      \nonumber  &\leq \left\|h_1-h_2\right\|\int  p_{h_\lambda}(\bbs,\bba)\left\|\zeta_t^{{h_1}}\right\|\sum_{r=0}^{t+u}\left\|\zeta_r^{h_\lambda}\right\| \,d\bbs d\bba.
\end{align}
      Writing the previous integral as an expectation and using the fact that $\zeta_t^{h_1}=\zeta_t^{h_\lambda}+\Sigma^{-1}(h_\lambda(s_t)-h_1(s_t)$, it reduces to 
      \begin{equation}
        I_1\leq
\left\|h_1-h_2\right\|\mathbb E_{p_{h_\lambda}}\hspace{-.1cm}\left[\left\|\zeta_t^{h_\lambda}\hspace{-.1cm}+\hspace{-.1cm}\Sigma^{-1}\hspace{-.1cm}(h_\lambda(s_t)-{h_1}(s_t))\right\|\sum_{r=0}^{t+u}\left\|\zeta_r^{h_\lambda}\right\|\right].\label{eq_I1lemma5}
      \end{equation}
Notice that $\Sigma^{-1/2}\zeta_t^{h_\lambda}$ and $\Sigma^{-1/2}\zeta_r^{h_\lambda}$ are multivariate independent white Gaussian variables, with first order moment bounded by $\sqrt{p}$.
      Then, the triangle inequality along with the bound \eqref{eqn_bound_hs} applied to $h(s)=h_{\lambda}(s)-h_1(s)$ in \eqref{eq_I1lemma5} yield
 \begin{align}
      \nonumber          &I_1\leq \left\|h_1-h_2\right\| \sum_{r=0}^{t+u} 
			\mathbb E_{p_{h_\lambda}}\left[ \left\|\zeta_t^{h_\lambda}\right\|\left\|\zeta_r^{h_\lambda}\right\|\right]\\
			 \nonumber&+\left\|h_1-h_2\right\| \sum_{r=0}^{t+u} \|h_\lambda-{h_1}\| \mathbb E_{p_{h_\lambda}}\left[ \left\|\Sigma^{-1} \zeta_r^{h_\lambda}\right\|\right]\\
	 &		\leq  (t+u+1)\left( \frac{p\left\|h_1-h_2\right\| }{\lambda_{\min}(\Sigma)}+ \frac{\sqrt{p}\left\|h_1-h_2\right\|^2}{\left(\lambda_{\min}(\Sigma)\right)^{3/2}}\right),
      \label{eqn_bound_first_integral}
     \end{align}

To bound $I_2$ in \eqref{eqn_split_delta},  apply again \eqref{eqn_bound_hs} to $h(s)=h_2(s)-h_1(s)$. It follows that $I_2$ is bounded by  $(\lambda_{\min}(\Sigma))^{-1}\|h_1-h_2\|$, which  together with \eqref{eqn_bound_first_integral} can be substituted in \eqref{eqn_split_delta}, to conclude the proof obtaining \eqref{eq_bound_lemma5} after adding the geometric sum
$\sum_{t=0}^\infty\sum_{u=0}^\infty (t+u+1)\gamma^{t+u} = {1+\gamma}/{(1-\gamma)^3}.$
 
\end{proof}

\begin{lemma}\label{lemma_stochastic_gradient}
Let $\Gamma(\cdot)$ be the Gamma function, and define 
  \begin{equation}
      \sigma= \frac{(3\gamma)^{1/3}}{(1-\gamma)^2}\frac{1}{\lambda_{\min}\left(\Sigma^{1/2}\right)}\left(4\frac{\Gamma(2+p/2)}{\Gamma(p/2)}\right)^{1/4},
  \end{equation}
  then the following bounds hold 
  \begin{equation}\label{eqn_moment_bounds}
\mathbb{E}\left[\left\|\hat{\nabla}_h U(h,\cdot) \right\|^2\right] \leq \sigma^2 \quad \mbox{and} \quad \mathbb{E}\left[\left\|\hat{\nabla}_h U(h,\cdot) \right\|^3\right] \leq \sigma^{3}.
  \end{equation}
\end{lemma}
\begin{proof}
Let us start by bounding the cube the norm of the stochastic gradient defined in \eqref{eqn_stochastic_grad}.
\begin{align}
\nonumber  \left\|  \hat{\nabla}_hU(h,\cdot)\right\|^3 &\leq \frac{1}{8(1-\gamma)^3}\left\|\hat{Q}(s_T,a_T;h)-\hat{Q}(s_T,\bar{a}_T;h)\right\|^3 \\
  &\times\left\|\kappa(s_T,\cdot)\right\|^3\left\|\Sigma^{-1}(a_T-h(s_T))\right\|^3.
\label{eq_bound_nablau}\end{align}
Substituting $\left\| \kappa(s_t,\cdot)\right\| = 1$ (cf., Definition \ref{def_rkhs}) and using the fact that the difference between estimates of $Q$ is bounded by $B_r (T_{Q}+T_Q^\prime)$, \eqref{eq_bound_nablau} is upper bounded by 
\begin{equation}\nonumber
\left\|  \hat{\nabla}_hU(h,\cdot)\right\|^3 \leq\frac{B_r^3}{8(1-\gamma)^3}(T_Q+T_Q^\prime)^3 \left\|\Sigma^{-1}(a_T-h(s_T))\right\|^3.
\end{equation}
From the independence of $T_Q$ and $T_Q^\prime$ with respect to the state evolution, and the monotonicity of the expectation, it results 
\begin{equation}\label{eqn_product_expectations}
  \begin{split}
  \mathbb{E}\left[\left\|  \hat{\nabla}_hU(h,\cdot)\right\|^3\right] \leq \frac{B_r^3}{8(1-\gamma)^3}\mathbb{E}\left[\left(T_Q+T_Q^\prime\right)^3\right]\\
    \times \mathbb{E}\left[\left\|\Sigma^{-1}(a_T-h(s_T))\right\|^3\right].
  \end{split}
\end{equation}
The sum of two independent geometric variables satisfies
\begin{equation}
\nonumber P(T_Q+T_Q^\prime=k) = (1-\gamma)^2(k+1) \gamma^k.
\end{equation} 
%
%
%
Thus, the third moment is upper bounded by
\begin{align}
    \nonumber \mathbb{E}\left[(T_Q+T_Q^\prime)^3\right] &= \sum_{k=0}^\infty k^3 (1-\gamma)^2(k+1)\gamma^k\\
		&= \frac{\gamma(1+14\gamma+8\gamma^2)}{(1-\gamma)^3}\leq \frac{23 \gamma}{(1-\gamma)^3}   \label{exp_sum_geom}   \end{align} 
where the last inequality follows from the fact that $\gamma<1$. On the other hand observe that $\left\|\Sigma^{-1/2}a_T-h(s_T)\right\|^2$ is Chi-squared with parameter $p$ since it is a sum of squares of normal random variables. Hence, the second expectation in \eqref{eqn_product_expectations} can be bounded using Jensen's inequality by, 
\begin{align}
   \nonumber &\mathbb{E}\left[\left\|\Sigma^{-1}(a_T-h(s_T))\right\|^3\right] \leq \frac{1}{\lambda_{\min}(\Sigma^{1/2})^{3}}\mathbb{E}\left[\chi_p^{3/2}\right]\\
     &\leq \frac{1}{\lambda_{\min}(\Sigma^{1/2})^{3}}\mathbb{E}\left[\chi_p^{2}\right]^{3/4} =\frac{1}{\lambda_{\min}(\Sigma^{1/2})^{3}}\left(4\frac{\Gamma(2+p/2)}{\Gamma(p/2)}\right)^{3/4}\label{bound_gaussian_cube}
  \end{align}
Substituting \eqref{exp_sum_geom} and \eqref{bound_gaussian_cube} in \eqref{eqn_product_expectations} yields the the bound for the  third moment of the stochastic gradient  in  \eqref{eqn_moment_bounds}. To validate the bound on the second moment and conclude the proof, consider $x=\left\|  \hat{\nabla}_hU(h,\cdot)\right\|^3$ and observe that since $x^{2/3}$ is  concave,  one can reverse Jensen's inequality  to obtain 
\begin{equation}\nonumber
\mathbb{E}\left[\left(\left\|  \hat{\nabla}_hU(h,\cdot)\right\|^3\right)^{2/3}\right]\leq  \mathbb{E}\left[\left\|  \hat{\nabla}_hU(h,\cdot)\right\|^3\right]^{2/3} \leq \left(\sigma^3\right)^{2/3}  . 
  \end{equation}
  . \end{proof}

\begin{lemma}\label{lemma_square_int_martingale}
Let $e_j=\hat{\nabla}_h U(h_j)-{\nabla}_h U(h_j)$ and let $\eta_j$ be such that it satisfies \eqref{eqn_stepsize}. Then, the sequence  
%
$S_k  :=\sum_{j=0}^k \eta_j e_j,$
%
  converges to a finite limit with probability one.
\end{lemma}
\begin{proof}
  By virtue of Theorem 5.4.9 \cite{durrett2010probability}, it suffices to show that $S_k$ is a square integrable martingale and that
  \begin{equation}\label{eqn_last_condition_square_integrable}
\lim_{n\to\infty} \sum_{m=1}^n \mathbb{E}\left[(S_m-S_{m-1})^2\big| \ccalF_{m}\right] <\infty \quad \mbox{a.e.}
  \end{equation}

Recall that the estimate of the gradient is unbiased, i.e. $\mathbb{E}\left[ \hat{\nabla}_h U(h_k,\cdot)\big|\ccalF_k\right]={\nabla}_h U(h_k,\cdot)$, hence we have that $\mathbb{E}\left[ e_k\big|\ccalF_k\right]=0$. This allows us to write 
            %
$\mathbb{E}\left[S_k\big|\ccalF_k\right] = S_{k-1} + \mathbb{E}\left[\eta_ke_k\big|\ccalF_k\right] = S_{k-1}.$
            %
            Thus $S_k$ is a martingale. To show that it is square integrable, compute the squared norm of $S_k$ as
            \begin{equation}\begin{split}
              \left\|S_k\right\|^2 = \left\|\sum_{j=0}^k \eta_je_j\right\|^2 = \eta_k^2\|e_k\|^2+2\eta_ke_k^\top\sum_{j=0}^{k-1}\eta_je_j \\+ \left\|\sum_{j=0}^{k-1} \eta_je_j\right\|^2 
              = \eta_k^2\|e_k\|^2 + 2\eta_k e_k^\top S_{k-1} + \left\|S_{k-1}\right\|^2.
              \end{split}\end{equation}
            Take the expectation with respect to the sigma field $\ccalF_k$ and use the fact that $\mathbb{E}\left[ e_k\big|\ccalF_k\right]=0$ to write
            \begin{equation}
              \mathbb{E}\left[\left\|S_k\right\|^2\big|\ccalF_k\right] = \eta_k^2\mathbb{E}\left[\|e_k\|^2\big|\ccalF_k\right]+\left\|S_{k-1}\right\|^2.
            \end{equation}
            %
                        %
												%
												%
           Taking expectations with respect to smaller sigma algebras recursively yields 
                          $ \mathbb{E}\left[\left\|S_k\right\|^2\right] = \sum_{j=0}^k\eta_j^2\mathbb{E}\left[\|e_j\|^2\right]$.
                        Since the step sizes are square summable and the second moment of the error is bounded (cf., lemmas \ref{lemma_bounded_gradient} and \ref{lemma_stochastic_gradient}) the second moment of $S_k$ is bounded for all $k$. We next show that \eqref{eqn_last_condition_square_integrable} holds. By definition of $S_k$ one has that 
                        \begin{equation}\begin{split}
\sum_{m=1}^n \mathbb{E}\left[\|S_m-S_{m-1}\|^2\big|\ccalF_{m}\right]= \sum_{m=1}^n \eta_m^2 \mathbb{E}\left[\|e_j\|^2\big|\ccalF_m\right].
\end{split}                          \end{equation}
                       Which is bounded for all $n$ as it was previously argued. This completes the proof that $\lim_{k\to\infty} S_k$ converges to a finite random variable with probability one.

\end{proof}

\bibliographystyle{ieeetr}
\bibliography{bib}
\vspace{-0.5cm}
\begin{IEEEbiography}[{\includegraphics[width=1in,height=1.25in,clip,keepaspectratio]{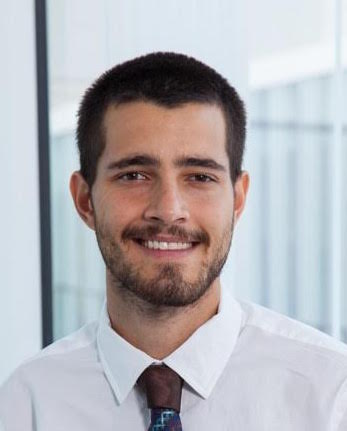}}]{Santiago Paternain} 
  received the B.Sc. degree in electrical engineering from Universidad de la Rep\'ublica Oriental del Uruguay, Montevideo, Uruguay in 2012. Since August 2013, he has been working toward the Ph.D. degree in the Department of Electrical and Systems Engineering, University of Pennsylvania. His research interests include optimization and control of dynamical systems. Mr. Paternain received the 2017 CDC best paper award.
\end{IEEEbiography}
\vspace{-.5cm}
\begin{IEEEbiography}[{\includegraphics[width=1in,height=1.25in,clip,keepaspectratio]{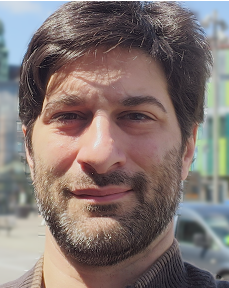}}]{Juan Andr\'es Bazerque} received the B.Sc. degree in electrical engineering from Universidad de la Rep\'ublica (UdelaR), Montevideo, Uruguay, in 2003, and the M.Sc. and Ph.D. degrees from the Department of Electrical and Computer Engineering, University of Minnesota (UofM), Minneapolis, in 2010 and 1013 respectively.
 Since 2015 he is an Assistant Professor with the  Department of Electrical Engineering  at UdelaR. His current research interests include  stochastic optimization and  networked systems, focusing on reinforcement learning, graph signal processing, and power systems optimization and control. 
Dr. Bazerque is the recipient of the UofM's Master Thesis Award 2009-2010, and co-reciepient of the best paper award at the 2nd International Conference on Cognitive Radio Oriented Wireless Networks and Communication 2007.   
 \end{IEEEbiography}
\vspace{-.5cm}
\begin{IEEEbiography}[{\includegraphics[width=1in,height=1.25in,clip,keepaspectratio]{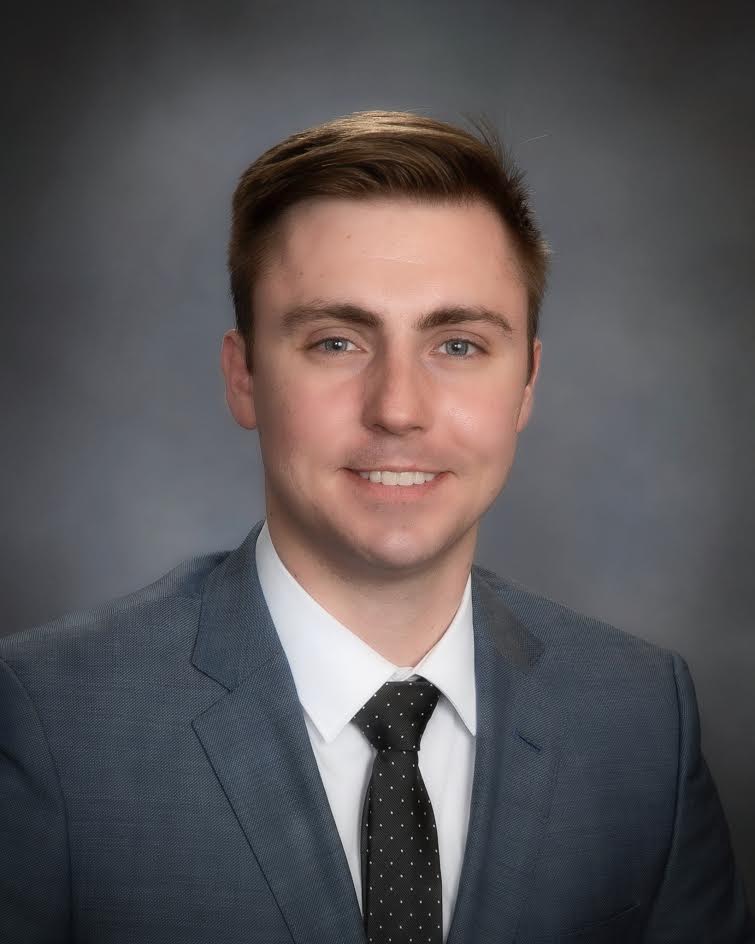}}]{Austin Small} Austin C. Small is currently pursuing the B.S. degree in computer science and electrical engineering with the University of Pennsylvania, Philadelphia, PA, USA.  His current research interests include machine learning for distributed systems, robotics, and applications at the intersection of machine learning and biomedical research.
\end{IEEEbiography}
\vspace{-.5cm}
\begin{IEEEbiography}[{\includegraphics[width=1in,height=1.25in,clip,keepaspectratio]{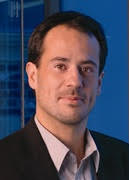}}]{Alejandro Ribeiro}  received the B.Sc. degree in electrical engineering from the Universidad de la Rep\'ublica Oriental del Uruguay, Montevideo, in 1998 and the M.Sc. and Ph.D. degree in electrical engineering from the Department of Electrical and Computer Engineering, the University of Minnesota, Minneapolis in 2005 and 2007. From 1998 to 2003, he was a member of the technical staff at Bellsouth Montevideo. After his M.Sc. and Ph.D studies, in 2008 he joined the University of Pennsylvania (Penn), Philadelphia, where he is currently the Rosenbluth Associate Professor at the Department of Electrical and Systems Engineering. His research interests are in the applications of statistical signal processing to the study of networks and networked phenomena. His focus is on structured representations of networked data structures, graph signal processing, network optimization, robot teams, and networked control. Dr. Ribeiro received the 2014 O. Hugo Schuck best paper award, and paper awards at the 2016 SSP Workshop, 2016 SAM Workshop, 2015 Asilomar SSC Conference, ACC 2013, ICASSP 2006, and ICASSP 2005. His teaching has been recognized with the 2017 Lindback award for distinguished teaching and the 2012 S. Reid Warren, Jr. Award presented by Penn's undergraduate student body for outstanding teaching. Dr. Ribeiro is a Fulbright scholar class of 2003 and a Penn Fellow class of 2015.
\end{IEEEbiography}



\end{document}